\documentclass[a4paper,UKenglish,cleveref, autoref, thm-restate]{lipics-v2019}

\newif\iflongversion
\longversionfalse

\usepackage{tikz}
\usetikzlibrary{calc}
\usetikzlibrary{patterns}
\usepackage{csquotes}

\usepackage{mdframed}

\crefname{observation}{observation}{observations}
\Crefname{observation}{Observation}{Observations}

\title{On the Complexity of Recoverable Robust Optimization in the Polynomial Hierarchy} 

\titlerunning{Complexity of Recoverable Robust Optimization} 

\author{Christoph Grüne}{Department of Computer Science, RWTH Aachen University, Germany}{gruene@algo.rwth-aachen.de}{https://orcid.org/0000-0002-7789-8870}{Funded by the German Research Foundation (DFG) – GRK 2236/2.}
\author{Lasse Wulf}{Section of Algorithms, Logic and Graphs, Technical University of Denmark, Kongens Lyngby, Denmark}{lawu@dtu.dk}{https://orcid.org/0000-0001-7139-4092}{Funded by the Carlsberg Foundation CF21-0302 ``Graph Algorithms with Geometric Applications''.}

\authorrunning{C. Grüne and L. Wulf} 

\Copyright{Christoph Grüne and Lasse Wulf} 

\ccsdesc[500]{Theory of computation~Problems, reductions and completeness} 

\keywords{Complexity, Robust Optimization, Recoverable Robust Optimization, Two-Stage Problems, Polynomial Hierarchy, Sigma 2, Sigma 3} 

\category{} 

\relatedversion{} 

\supplement{}

\nolinenumbers 

\hideLIPIcs  

\EventEditors{John Q. Open and Joan R. Access}
\EventNoEds{2}
\EventLongTitle{42nd Conference on Very Important Topics (CVIT 2016)}
\EventShortTitle{CVIT 2016}
\EventAcronym{CVIT}
\EventYear{2016}
\EventDate{December 24--27, 2016}
\EventLocation{Little Whinging, United Kingdom}
\EventLogo{}
\SeriesVolume{42}
\ArticleNo{23}

\definecolor{darkgreen}{RGB}{0,128,0}
\definecolor{darkred}{RGB}{128,0,0}

\newcommand{\R}{\mathbb{R}}
\newcommand{\N}{\mathbb{N}}
\newcommand{\Z}{\mathbb{Z}}
\newcommand{\I}{\mathcal{I}}
\newcommand{\U}{\mathcal{U}}
\newcommand{\F}{\mathcal{F}}
\newcommand{\sol}{\mathcal{S}}
\newcommand{\powerset}[1]{2^{#1}}
\newcommand{\set}[1]{\{ #1 \}}
\newcommand{\fromto}[2]{\set{#1, \ldots, #2}}

\DeclareMathOperator{\poly}{poly}

\DeclareMathOperator{\dist}{dist}
\newcommand{\leqSSP}{\leq_\text{SSP}}
\newcommand{\bin}{\set{0,1}}

\newcommand{\NP}{\textsl{NP}}

\begin{document}

\maketitle
\begin{abstract}
    Recoverable robust optimization is a popular multi-stage approach, in which it is possible to adjust a first-stage solution after the uncertain cost scenario is revealed. 
We consider recoverable robust optimization in combination with discrete budgeted uncertainty. 
In this setting, it seems plausible that many problems become $\Sigma^p_3$-complete and therefore it is impossible to find compact IP formulations of them (unless the unlikely conjecture {\it NP} $= \Sigma^p_3$ holds).
Even though this seems plausible, few concrete results of this kind are known.
In this paper, we fill that gap of knowledge.
We consider recoverable robust optimization for the nominal problems of Sat, 3Sat, vertex cover, dominating set, set cover, hitting set, feedback vertex set, feedback arc set, uncapacitated facility location, $p$-center, $p$-median, independent set, clique, subset sum, knapsack, partition, scheduling, Hamiltonian path/cycle (directed/undirected), TSP, $k$-disjoint path ($k \geq 2$), and Steiner tree. 
We show that for each of these problems, and for each of three widely used distance measures, the recoverable robust problem becomes $\Sigma^p_3$-complete.
Concretely, we show that all these problems share a certain abstract property and prove that this property implies that their robust recoverable counterpart is $\Sigma^p_3$-complete. 
This reveals the insight that all the above problems are $\Sigma^p_3$-complete \enquote{for the same reason}. Our result extends a recent framework by Grüne and Wulf.    
\end{abstract}

\newpage

\section{Introduction}
Real-world decision makers are faced with a large degree of uncertainty when making important decisions in economics, planning, and operations research. 
The successful area of \emph{Robust optimization} \cite{DBLP:books/degruyter/Ben-TalGN09,robook,kouvelis2013robust} has been developed as one possible way to deal with these uncertainties. 
However, sometimes the classical robust optimization approach has turned out to be too conservative. For this reason, \emph{Recoverable robust optimization} was introduced by Liebchen, Lübbecke, Möhring and Stiller \cite{DBLP:series/lncs/LiebchenLMS09}. 
Initially motivated by train scheduling problems, it has since then found wide-spread application in practice and in the analysis of standard optimization problems.
For instance, it has been successfully applied to combinatorial problems such as
s-t-path \cite{busing2011recoverable,DBLP:journals/corr/abs-2403-20000,DBLP:journals/corr/abs-2401-05715},
matching \cite{DBLP:journals/networks/DouradoMPRPA15},
scheduling \cite{DBLP:journals/dam/BoldG22},
spanning tree \cite{DBLP:journals/jco/HradovichKZ17,DBLP:journals/ol/HradovichKZ17},
knapsack \cite{DBLP:journals/ejco/BusingGKK19,DBLP:conf/inoc/BusingKK11,DBLP:journals/ol/BusingKK11,DBLP:journals/dam/LachmannLW21}, and
TSP \cite{DBLP:journals/ol/ChasseinG16,DBLP:journals/eor/GoerigkLW22a}.
Recoverable robust optimization follows a two-step robust optimization approach.
In contrast to classic robust optimization, where a first-stage solution cannot be changed, 
in recoverable robust optimization the decision maker is allowed to incorporate a limited recovery action after the underlying uncertainty is revealed. 
Mathematically, this is described by the expression
\[
\min_{S_1} \quad \max_{c_2 \in C} \quad \min_{S_2, \dist(S_1, S_2) \leq \kappa} c_1(S_1) + c_2(S_2).
\]
Here, $S_1$ denotes the first-stage solution, $S_2$ denotes the second-stage solution, $C$ denotes the set of uncertain scenarios, and {\it dist} is some distance function. 
A more formal definition will be given later.

In this paper, we make the standard assumption that the underlying uncertainty for the robust problem is given as a discrete budgeted uncertainty set, also denoted as discrete $\Gamma$-uncertainty \cite{DBLP:journals/mp/BertsimasS03}.

\textbf{The study of $\Sigma^p_3$-complete problems.}
Recoverable robust problems are described by min-max-min formulations. 
From a complexity-theoretic point of view, such min-max-min problems are often even harder than NP-complete, 
namely they are often complete for the third stage of the so-called \emph{polynomial hierarchy} \cite{DBLP:journals/tcs/Stockmeyer76}. 
A problem complete for the $k$-th stage of the hierarchy is called $\Sigma^p_k$-complete. (In this paper, we are concerned mostly with the case of $k=3$). 
The theoretical study of $\Sigma^p_k$-complete problems is important: 
If a problem is found to be $\Sigma^p_3$-complete, it means that, under some basic complexity-theoretic assumptions\footnote{More specifically, we assume here that $\Sigma^p_3 \neq \text{NP}$. 
This is a similar assumption to the famous $\text{P} \neq \text{NP}$ assumption, and it is believed to be similarly unlikely. 
However, the true status of the conjecture is not known. 
More details can be found in \cite{DBLP:journals/4or/Woeginger21}.}, 
it is not possible to find a integer programming formulation of the problem of polynomial size \cite{DBLP:journals/4or/Woeginger21} (also called a \emph{compact} model). 
This means that no matter how cleverly a decision maker tries to design their integer programming model, 
it must inherently have a huge number of constraints and/or variables, and may be much harder to solve than even NP-complete problems.

Even though this fact makes the study of $\Sigma^p_3$-complete problems compelling, surprisingly few results relevant to the area of min-max-min combinatorial optimization were known until recently. 
While the usual approach to prove $\Sigma^p_k$-completeness (or NP-completeness) is to formulate a new proof for each single problem,
a recent paper by Grüne \& Wulf \cite{DBLP:journals/corr/abs-2311-10540}, extending earlier ideas by Johannes \cite{johannes2011new} breaks with this approach. 
Instead, it is shown that there exists a large 'base list' of problems (called \emph{SSP-NP-complete} problems in \cite{DBLP:journals/corr/abs-2311-10540}), including many classic problems like independent set, vertex cover, knapsack, TSP, etc.
Grüne and Wulf show that for each problem from the base list, some corresponding min-max version is $\Sigma^p_2$-complete, and some corresponding min-max-min version is $\Sigma^p_3$-complete. 
This approach has three main advantages: 
A1.) It uncovers a large number of previously unknown $\Sigma^p_k$-complete problems. 
A2.) It reveals the theoretically interesting insight, that for all these problems the $\Sigma^p_k$-completeness follows from essentially the same argument. 
A3.) It can simplify future proofs, since heuristically it seems to be true that for a new problem it is often easier to show that the nominal problem belongs to the list of SSP-NP-complete problems, than to find a $\Sigma^p_k$-completeness proof from scratch.

The main goal of the current paper is to extend the framework of Grüne \& Wulf \cite{DBLP:journals/corr/abs-2311-10540} also to the setting of recoverable optimization, 
since this setting was not considered in the original paper.
It turns out that compared to \cite{DBLP:journals/corr/abs-2311-10540} more complicated assumptions are needed. We introduce a set of sufficient conditions for some nominal problem, which imply that the recoverable robust version of that problem is $\Sigma^p_3$-complete.

\textbf{Our results.}
We consider recoverable robust optimization for the following nominal problems: Sat, 3Sat, vertex cover, dominating set, set cover, hitting set, feedback vertex set, feedback arc set, uncapacitated facility location, $p$-center, $p$-median, independent set, clique, subset sum, knapsack, partition, scheduling, Hamiltonian path/cycle (directed/undirected), TSP, $k$-disjoint path ($k \geq 2$), and Steiner tree.
In addition we consider the three most popular  distance measures \emph{dist} used in recoverable robust optimization: The $\kappa$-addition distance, the $\kappa$-deletion distance, and the Hamming distance. (A formal definition of these distance measures is given in \cref{subsec:distance-mesures})

We show that for every combination of the above problems with any of the three distance measures, the recoverable robust problem (with discrete budgeted uncertainty) is $\Sigma^p_3$-complete.
More generally, we identify an abstract property of all our studied problems, and prove that this abstract property already implies that the recoverable robust problem becomes $\Sigma^p_3$-complete. This answers a question asked by Goerigk, Lendl and Wulf \cite{DBLP:journals/dam/GoerigkLW24}. As a consequence, the advantages A1--A3 as explained above also apply to recoverable robust problems.
We remark that $\Sigma^p_3$-completeness was already known in the case of clique/independent set, TSP or shortest path combined with the Hamming distance \cite{DBLP:journals/dam/GoerigkLW24}. It was also already known for shortest path in combination with all three distance measures \cite{DBLP:journals/corr/abs-2403-20000} and for several optimization problems \cite{DBLP:conf/latin/Grune24} for so-called xor-dependencies and $\Gamma$-set scenarios.
Hence our work is an extension of these results.

\subsection{Related Work}
Recoverable robustness concepts were analyzed for a variety of different standard optimization problems.
Büsing \cite{busing2011recoverable} analyzed the recoverable shortest path problem with discrete budgeted uncertainty, in which adding at most $\kappa$ elements to the second stage solution are allowed.
This analysis was lately continued by Jackiewicz, Kasperski and Zieliński \cite{DBLP:journals/corr/abs-2401-05715} for several different graph classes and for interval budgeted uncertainty.
Furthermore, recoverable robust knapsack was analyzed by Büsing, Koster and Kutschka \cite{DBLP:journals/ol/BusingKK11} for discrete scenarios and in which adding at most $\kappa$ elements and deleting at most $\ell$ elements are allowed.
Büsing, Koster and Kutschka \cite{DBLP:conf/inoc/BusingKK11} also analyzed the $\Gamma$-scenario case (discrete budgeted uncertainty) while allowing at most $\ell$ elements to be deleted.
This work was further continued by Büsing, Goderbauer, Koster and Kutschka \cite{DBLP:journals/ejco/BusingGKK19}.
Further classical optimization problems that were studied in the recoverable robustness context are matching by Dourado, Meierling, Penso, Rautenbach, Protti and de Almeida \cite{DBLP:journals/networks/DouradoMPRPA15} and spanning tree under interval cost uncertainty by Hradovich, Kasperski, Zieliński \cite{DBLP:journals/jco/HradovichKZ17,DBLP:journals/ol/HradovichKZ17}.
Additionally, Lendl, Peis, and Timmermans \cite{DBLP:journals/mp/LendlPT22} examined matroidal problems.
One variant of independent set for recoverable robustness with a commitment property was analyzed by Hommelsheim, Megow, Muluk and Peis \cite{DBLP:journals/corr/abs-2306-08546}.
Beside these problems, recoverable robust selection was studied by Kasperski and Zieliński \cite{DBLP:journals/dam/KasperskiZ17}, Chassein, Goerigk, Kasperski, and Zieliński \cite{DBLP:journals/eor/ChasseinGKZ18}, and Goerigk, Lendl and Wulf \cite{DBLP:journals/eor/GoerigkLW22a}.
Moreover, Lachmann, Lendl and Woeginger \cite{DBLP:journals/dam/LachmannLW21} developed a linear time algorithm for the recoverable $\Gamma$-robust knapsack problem.
The recoverable robust assignment problem was also investigated by Fischer, Hartmann, Lendl and Woeginger \cite{DBLP:conf/iwpec/0001HLW21}.

Approximation results for recoverable robust problems were also achieved.
For example, the recoverable traveling salesman problem was studied by Chassein and Goerigk \cite{DBLP:journals/ol/ChasseinG16}, as well as Goerigk, Lendl, Wulf \cite{DBLP:journals/eor/GoerigkLW22a}, where the latter showed a 4-approximation algorithm.
Beyond that Bold and Goerigk \cite{DBLP:journals/dam/BoldG22} presented a 2-approximation for the recoverable robust single machine scheduling problem under interval uncertainty.

Closely related to our work are the complexity studies by Goerigk, Lendl and Wulf \cite{DBLP:journals/dam/GoerigkLW24} who analyze the problems independent set, traveling salesman and vertex cover and obtain $\Sigma^p_3$-completeness for the three problems.
At the same time, Grüne \cite{DBLP:conf/latin/Grune24} introduced a gadget reduction framework to derive $\Sigma^p_3$-completeness for various different optimization problems by reusing already known reductions for recoverable robust problems with so-called xor-dependency scenarios and $\Gamma$-set scenarios.
Additionally, Jackiewicz, Kasperski and Zieliński \cite{DBLP:journals/corr/abs-2403-20000} show that the shortest path problem with discrete budgeted interval uncertainty is $\Sigma^p_3$-complete.

\section{Preliminaries}
\label{sec:prelim}
A \emph{language} is a set $L\subseteq \bin^*$.
A language $L$ is contained in $\Sigma^p_k$ iff there exists some polynomial-time computable function $V$ (verifier), and $m_1,m_2,\ldots, m_k = \poly(|w|)$ such that for all $w \in \set{0,1}^*$
\[
    w \in L \ \Leftrightarrow \ \exists y_1 \in \set{0,1}^{m_1} \ \forall y_2 \in \set{0,1}^{m_2} \ldots \ Q y_k \in \set{0,1}^{m_k}: V(w,y_1,y_2,\ldots,y_k) = 1,
\]
where $Q = \exists$, if $k$ is odd, and $Q = \forall$, if $k$ even.
An introduction to the polynomial hierarchy and the classes $\Sigma^p_k$ can be found in the book by Papadimitriou \cite{DBLP:books/daglib/0072413} or in the article by Jeroslow \cite{DBLP:journals/mp/Jeroslow85}.
An introduction specifically in the context of bilevel optimization can be found in the article of Woeginger \cite{DBLP:journals/4or/Woeginger21}.
A \emph{many-one-reduction} or \emph{Karp-reduction} from a language $L$ to a language $L'$ is a map $f : \bin^* \to \bin^*$ such that $w \in L$ iff $f(w) \in L'$ for all $w \in \bin^*$. 
A language $L$ is $\Sigma^p_k$-hard, if every $L' \in \Sigma^p_k$ can be reduced to $L$ with a polynomial-time many-one reduction. If $L$ is both $\Sigma^p_k$-hard and contained in $\Sigma^p_k$, it is $\Sigma^p_k$-complete.

For some cost function $c : U \to \R$, and some subset $U' \subseteq U$, we define the cost of the subset $U'$ as $c(U') := \sum_{u \in U'} c(u)$. For a map $f : A \to B$ and some subset $A' \subseteq A$, we define the image of the subset $A'$ as $f(A') = \set{f(a) : a \in A'}$. 

\section{Framework}
\label{sec:framework}

Because this works builds upon the framework developed by Grüne and Wulf, it is essential to reintroduce its key concepts.
A more comprehensive explanation of these concepts and their underlying motivation can be found in the original paper \cite{DBLP:journals/corr/abs-2311-10540}.
Grüne and Wulf begin by providing a precise definition of their primary focus: linear optimization problems, or in short LOP problems.
An example of an LOP problem is vertex cover.

\begin{definition}[Linear Optimization Problem, from \cite{DBLP:journals/corr/abs-2311-10540}]
\label{def:LOSPP}
    A linear optimization problem (or in short LOP problem)  $\Pi$ is a tuple $(\I, \U, \F, d, t)$, such that
    \begin{itemize}
        \item $\I \subseteq \{0,1\}^*$ is a language. We call $\I$ the set of instances of $\Pi$.
        \item For each instance $I \in \I$, there is some
        \begin{itemize}
            \item set $\U(I)$ which we call the universe associated to the instance $I$.
            \item set $\F(I) \subseteq \powerset{\U(I)}$ that we call the feasible solution set associated to the instance $I$. 
            \item function $d^{(I)}: \U(I) \rightarrow \Z$ mapping each universe element $e$ to its costs $d^{(I)}(e)$.
            \item threshold $t^{(I)} \in \Z$. 
        \end{itemize}
    \end{itemize}
    For $I \in \I$, we define the solution set $\sol(I) := \set{S \in \F(I) : d^{(I)}(S) \leq t^{(I)}}$ as the set of feasible solutions below the cost threshold. 
    The instance $I$ is a Yes-instance, if and only if $\sol(I) \neq \emptyset$.
    We assume (for LOP problems in NP) that it can be checked in polynomial time in $|I|$ whether some proposed set $F \subseteq \U(I)$ is feasible.
\end{definition}

\begin{description}
    \item[]\textsc{Vertex Cover}\hfill\\
    \textbf{Instances:} Graph $G = (V, E)$, number $k \in \N$.\\
    \textbf{Universe:} Vertex set $V =: \U$.\\
    \textbf{Feasible solution set:} The set of all vertex covers of $G$.\\
    \textbf{Solution set:} The set of all vertex covers of $G$ of size at most $k$.
\end{description}

It turns out that often times the mathematical discussion is a lot clearer, when one omits the concepts $\F, d^{(I)}$, and $t^{(I)}$, since for the abstract proof of the theorems only $\I, \U, \sol$ are important. This leads to the following abstraction from the concept of an LOP problem:

\begin{definition}[Subset Search Problem (SSP), from \cite{DBLP:journals/corr/abs-2311-10540}]
\label{def:SSP}
A subset search problem (or short SSP problem) $\Pi$ is a tuple $(\I, \U, \sol)$, such that
\begin{itemize}
    \item $\I \subseteq \set{0,1}^*$ is a language. We call $\I$ the set of instances of $\Pi$. 
    \item For each instance $I \in \I$, there is some set $\U(I)$ which we call the universe associated to the instance $I$. 
    \item For each instance $I \in \I$, there is some (potentially empty) set $\sol(I)\subseteq \powerset{\U(I)}$ which we call the solution set associated to the instance $I$.
\end{itemize}
\end{definition}

An instance of an SSP problem is referred to as a yes-instance if $\sol(I) \neq \emptyset$.
Any LOP problem is transformable into an SSP problem by defining $\sol(I) := \set{S \in \F(I) : d^{(I)}(S) \leq t^{(I)}}$.
We refer to this as the \emph{SSP problem derived from an LOP problem}.
There are problems that are more naturally represented as SSP problems, rather than as LOP problems.
For instance, the problem \textsc{Satisfiability} can be derived as an SSP problem as follows.

\begin{description}
    \item[]\textsc{Satisfiability}\hfill\\
    \textbf{Instances:} Literal set $L = \fromto{\ell_1}{\ell_n} \cup \fromto{\overline \ell_1}{\overline \ell_n}$, clause set $C = \fromto{C_1}{C_m}$ such that $C_j \subseteq L$ for all $j \in \fromto{1}{m}.$\\
    \textbf{Universe:} $L =: \U$.\\
    \textbf{Solution set:} The set of all subsets $L' \subseteq \U$ of the literals such that for all $i \in \fromto{1}{n}$ we have $|L' \cap \set{\ell_i, \overline \ell_i}| = 1$, and such that $|L' \cap C_j| \geq 1$ for all clauses $C_j \in C$.
\end{description}

Grüne and Wulf introduce a novel type of reduction, termed SSP reduction.
In essence, a typical polynomial-time reduction from a problem $\Pi$ to another problem $\Pi'$ possesses the SSP property if it includes an additional injective mapping $f$ that embeds the universe $\U$ elements of $\Pi$ into the universe $\U'$ of $\Pi'$.
This embedding allows $\Pi$ to be viewed as a 'subinstance' of $\Pi'$ while preserving the topology of solutions within the subset induced by the image of $f$.
More precisely, we interpret $W$ as the subinstance of $\Pi$ within the instance of $\Pi'$ and we require the following two conditions to be met:
\begin{enumerate}
    \item For every solution $S'$ of $\Pi'$, the set $f^{-1}(S' \cap W)$ is a solution of $\Pi$.
    \item For every solution $S$ of $\Pi$, the set $f(S)$ is a partial solution of $\Pi'$ and can be extended to a full solution using elements that are not in $W$.
\end{enumerate}
These two conditions are encapsulated in the single equation (\ref{eq:SSP}).
We indicate that such a reduction exists by $\Pi \leq_\text{SSP} \Pi'$.
For a more intuitive explanation and an example that shows \textsc{3Sat} $\leq_\text{SSP}$ \textsc{Vertex Cover}, we direct readers to \cite{DBLP:journals/corr/abs-2311-10540}.

\begin{definition}[SSP Reduction, from \cite{DBLP:journals/corr/abs-2311-10540}]
\label{def:ssp-reduction}
    Let $\Pi = (\I,\U,\sol)$ and $\Pi' = (\I',\U',\sol')$ be two SSP problems. We say that there is an SSP reduction from $\Pi$ to $\Pi'$, and write $\Pi \leqSSP \Pi'$, if
    \begin{itemize}
        \item There exists a function $g : \{0,1\}^* \to \{0,1\}^*$ computable in polynomial time in the input size $|I|$, such that $I$ is a Yes-instance iff $g(I)$ is a Yes-instance (i.e. $\sol(I) \neq \emptyset$ iff $\sol'(g(I)) \neq \emptyset$).
        \item There exist functions $(f_I)_{I \in \I}$ computable in polynomial time in $|I|$ such that for all instances $I \in \I$, we have that $f_I : \U(I) \to \U'(g(I))$ is an injective function mapping from the universe of the instance $I$ to the universe of the instance $g(I)$ such that 
        \begin{equation}
            \set{f_I(S) : S \in \sol(I) } = \set{S' \cap f_I(\U(I)) : S' \in  \sol'(g(I))}. \label{eq:SSP}
        \end{equation}

    \end{itemize}
\end{definition}

In \cite{DBLP:journals/corr/abs-2311-10540}, it is demonstrated that SSP reductions are transitive, that is $\Pi_1 \leqSSP \Pi_2$ and $\Pi_2 \leqSSP \Pi_3$ imply $\Pi_1 \leqSSP \Pi_3$.
We denote the class of SSP-NP-complete problems by SSP-NPc.
It includes all SSP problems $\Pi$ that are polynomial-time verifiable and there is an SSP reduction from \textsc{Satisfiability}, i.e. $\textsc{Satisfiability} \leq_\text{SSP} \Pi$.
The key insight in \cite{DBLP:journals/corr/abs-2311-10540} is that many classical problems fall within the class SSP-NPc.
Additionally, this insight can be leveraged to show that their corresponding min-max version are $\Sigma^p_2$-complete.

\section{Recoverable Robust Problems}
In this section, we consider recoverable robust optimization problems. We show that the recoverable robust optimization problem is $\Sigma^p_3$-complete for the following nominal problems:
satisfiability,
3-satisfiability,
vertex cover,
dominating set,
set cover,
hitting set,
feedback vertex set,
feedback arc set,
uncapacitated facility location,
p-center,
p-median,
independent set,
clique,
subset sum,
knapsack,
partition,
scheduling,
(un)directed Hamiltonian path,
(un)directed Hamiltonian cycle,
traveling salesman,
two directed disjoint path,
k directed disjoint path,
and Steiner tree.

Recoverable robust optimization problems are defined as follows: 
We are given some instance $I$ of a linear optimization problem (like the shortest path problem, the traveling salesman problem, etc.), and are faced with an uncertain future. 
The goal is to find a feasible solution $S_1$ for the first stage (called here-and-now decision) such that after the reveal of the uncertainty we can find a feasible solution $S_2$ in the second stage (called wait-and-see decision) 
such that $S_1$ and $S_2$ are not far away from each other according to some distance measure. 
Formally we require $\dist(S_1, S_2) \leq \kappa$, where $\dist(S_1, S_2)$ is some abstract distance function (for example the Hamming distance $|S_1 \vartriangle S_2|$).
The cost of the solution is given by $c_1(S_1) + c_2(S_2)$.
Here, $c_1$ is a fixed cost function not affected by uncertainty, called the setup costs, and $c_2$ is affected by uncertainty. 
More specifically, we assume that $c_2 \in C_\Gamma$, that is, $c_2$ is affected by discrete budgeted uncertainty\footnote{Usually in literature, this set is denoted by $\U_\Gamma$, however, we have already used the letter $\U$ in this paper.} $C_\Gamma$. 
Precisely, given $\Gamma \in \N$ and upper and lower bounds $\underline c(u) \leq \overline c(u)$ for all elements in the universe, the set $C_\Gamma$ contains all cost functions such that at most $\Gamma$ elements $u \in \U(I)$ have costs of $\overline c(u)$, while all other have $\underline c(u)$:
    \[
        C_\Gamma := \{c_2 \mid \forall u \in \U(I) : c_2(u) = \underline c(u) + \delta_u (\overline c(u) - \underline c(u)), \ \delta_u \in \{0,1\}, \sum_{u \in \U(I)} \delta_u \leq \Gamma\}
    \]

This leads to the following abstract definition of a recoverable robust problem.
We remark that the formal definition in some sense \enquote{disregards} the cost functions $d$ and the cost threshold $t$ of the original LOP problem.
This is intentional because in the new instance of the recoverable robust problem, it becomes necessary to substitute the old function $d$ and the old threshold $t$ by new ones.
However, these concepts are necessary to correctly understand the proofs.

\begin{definition}[Recoverable Robust Problem]
\label{def:RecoverableProblem}
    Let an LOP problem $\Pi = (\I, \U, \F, d, t)$ and a distance measure $\text{dist}: 2^\U \times 2^\U \rightarrow \mathbb{R}_{\geq 0}$ be given.
    The recoverable robust problem associated to $\Pi$ is denoted by $\textsc{RR-}\Pi$ and defined as follows:
    The input is an
    instance $I \in \I$ together with
    three cost functions
    $c_1 : \U(I) \rightarrow \Z, \ \underline c: \U(I) \rightarrow \Z$ and
    $\overline c: \U(I) \rightarrow \Z$ and
    a cost threshold $t_{\textit{RR}} \in \Z$,
    an uncertainty parameter $\Gamma \in \N_0$ and
    a recoverability parameter $\kappa \in \N_0$.
    The question is whether
    \begin{align}
        \min_{S_1 \in \F(I)} \max_{c_2 \in C_\Gamma} \min_{\substack{S_2 \in \F(I) \\ \dist(S_1, S_2)  \leq \kappa}} \ c_1(S_1) + c_2(S_2) \leq t_{\textit{RR}}. \label{eq:recoverable-definition}
    \end{align}
\end{definition}

\textbf{Example.}
Let $\Pi =$ TSP.
According to \Cref{apx:ssp-reduction:tsp} the TSP is encoded as LOP problem the following way: An instance is given by $I = (G, d, t)$, where $G = (V, E)$ is a complete undirected graph, $d : E \to \N_0$ are the edge costs and $t$ is the cost threshold.
The decision problem of TSP asks if there is a tour $T \subseteq E$ with cost $d(T) \leq t$.
The universe is $\U(I) = E$. The set $\F(I) \subseteq 2^E$ is the set of all feasible tours (including those of cost greater than $t$). The set $\sol(I)$ is the set of all tours of cost at most $t$.
To turn the TSP into a recoverable robust problem, we \enquote{forget} about the cost function $d$ and the threshold $t$.
Given $c_1, \underline c, \overline c, \Gamma, \kappa, t_{\textit{RR}}$, the decision problem associated to the recoverable robust TSP is to decide whether \Cref{eq:recoverable-definition} holds.

We remark that this definition does not include all SSP problems but only LOP problems. 
This is for the reason that recoverable robust optimization is usually considered only for linear optimization problems, which distinguish between feasible solutions $\F(I)$ and optimal solutions $\sol(I)$.
In contrast, recoverable robust optimization is usually not considered for pure feasibility problems, like \textsc{Sat}, which are modelled as SSP problems in our framework. 
However, it is still possible to define a variant of recoverable robust optimization that is applicable to \emph{all} SSP problems. 
Indeed, it turns out that such a definition becomes helpful for our proof. 
Hence, in \cref{subsec:comb-recov-definition}, we introduce a corresponding definition (and also show $\Sigma^p_3$-completeness of several problems with this new definition).

\subsection{Distance Measures}
\label{subsec:distance-mesures}
Recoverable robust problems require a distance measure to model the constraint that the solutions remain close to each other.
However, there are numerous problems that also have a different structures such that it is not possible to define distance measures for all possible types of recoverable robust problems.
Because we restrict ourselves to a certain kind of problems, namely subset search problems, we consider distance measures defined over sets.
This allows us to determine a specific enough definition of distance measure to be meaningful.
Furthermore, our definition needs to be general enough to include the distance measures used in the literature.
Among those are the following:
\begin{itemize}
    \item the \textit{$\kappa$-addition} or simply \textit{$\kappa$-distance measure} is used in \cite{DBLP:journals/networks/Busing12, DBLP:journals/jco/HradovichKZ17}: $\dist(A_1, A_2) = |A_2 \setminus A_1|$
    \item the \textit{$\kappa$-deletion distance} is used in \cite{DBLP:conf/inoc/BusingKK11}: $\dist(A_1, A_2) = |A_1 \setminus A_2|$
    \item the \textit{Hamming distance} measure is used in \cite{DBLP:journals/networks/DouradoMPRPA15,DBLP:journals/dam/GoerigkLW24,DBLP:conf/latin/Grune24}: $\dist(A_1, A_2) = |A_1 \vartriangle A_2|$
\end{itemize}

\begin{definition}[Distance Measure]
    Let $U$ be a set and $A_1, A_2 \subseteq U$ subsets.
    A distance measure on set $U$ is a map $\text{dist}_U: 2^U \times 2^U \rightarrow \R_{\geq 0}$ that adheres to the following properties
    \begin{itemize}
        \item computable in polynomial time 
        \item invariant on injective mappings $f: U \rightarrow U'$, i.e. $\text{dist}_U(A_1, A_2) = \text{dist}_{U'}(f(A_1), f(A_2))$,
        \item invariant on union, i.e. $\text{dist}_U(A_1, A_2) = \text{dist}_U(A_1 \cup \{x\}, A_2 \cup \{x\})$ for $x \in U \setminus (A_1 \cup A_2)$.
        \item $\text{dist}_U(A_1, A_1) = 0$
    \end{itemize}
    If $U$ is clear from the context, we omit the subscript.
\end{definition}

One can easily verify that all of the distance measures from above fulfill these criteria.

\subsection{Containment in $\Sigma^p_3$}
Recoverable robust problems can also be understood as a three-stage two-player game of an $\exists$-player playing against an adversary (the $\forall$-player).
The $\exists$-player controls both $\min$ operators and is able to choose the solutions $S_1$ and $S_2$.
On the other hand, the adversary controls the $\max$ operator and is able to choose the uncertainty scenario, i.e. the cost function $c_2$ from the set $C_\Gamma$.
Thus, it is possible to reformulate the question to
\[
    \exists S_1 \subseteq \U(I) : \forall c_2 \in C_\Gamma : \exists S_2 \subseteq \U(I) : S_1, S_2 \in \F(I), c_1(S_1) + c_2(S_2) \leq t \ \text{and} \ \dist(S_1, S_2)  \leq \kappa.
\]

With the game-theoretical perspective, it is intuitive to see the containment in $\Sigma^p_3$ for all problems that have a polynomial-time verifier.

\begin{theorem}\label{thm:hddr-containment}
    If $\Pi = (\I, \U, \sol)$ is an LOP problem in NP, then \textsc{RR-$\Pi$} is in $\Sigma^P_3$.
\end{theorem}
\begin{proof}
    We provide a polynomial-time algorithm $V$ that verifies a specific solution for the three quantifiers $y_1, y_2, y_3$ of polynomial size for instance $I$ such that
    \[
        I \in L \Leftrightarrow \exists y_1 \in \{0,1\}^{m_1} \forall y_2 \in \{0,1\}^{m_2} y_3 \in \{0,1\}^{m_3} : V(I, y_1, y_2, y_3) = 1.
    \]
    With the first $\exists$-quantified $y_1$, we encode the first solution $S_1 \subseteq \U(I)$.
    Because the universe is part of the input of the problem, this is at most linear in the input.
    Next, we encode all cost functions $c_2 \in C_\Gamma$ with the $\forall$-quantified $y_2$.
    For this, we encode which of the at most $\Gamma$ elements are chosen to have costs $\overline c$.
    This is at most linear in the input size.
    At last, we encode the second solution $S_2 \subseteq \U(I)$ with the help of the second $\exists$-quantified $y_3$.
    Again, this is at most linear in the input size as before.
    Now, we have to check whether $\dist(S_1, S_2) \leq \kappa$, $S_1, S_2 \in \F(I)$ and $c_1(S_1) + c_2(S_2) \leq t$.
    All of these checks can be done in polynomial time, where $S_1, S_2 \in \F(I)$ and $c_1(S_1) + c_2(S_2) \leq t$ can be checked in polynomial time because $\Pi \in \NP$.
    Furthermore, $\dist(S_1, S_2) \leq \kappa$ can be checked by our assumption that the abstract distance function is computable in polynomial time.
    It follows that $\textsc{RR-}\Pi$ is in $\Sigma^p_3$.
\end{proof}

\subsection{Combinatorial Recoverable Robust Problems}
\label{subsec:comb-recov-definition}

To show the hardness of recoverable robust problems, we introduce a new problem which we call the \emph{combinatorial version} of recoverable robust problems.
There are two differences to \Cref{def:RecoverableProblem}: In this new problem, the cost function is substituted by a set $B$ of so-called \emph{blockable elements} (this can be interpreted as the case where cost coefficients $\underline c, \overline c$ are restricted to come from $\set{0, \infty}$). Furthermore, in this new definition we substitute $\F(I)$ by $\sol(I)$.
As we show later, the hardness of the new combinatorial version also implies the hardness of the cost version.
Finally, as explained above, \Cref{def:CombinatorialRecoverableRobustProblem} applies to all SSP problems. 

\begin{definition}[Combinatorial Recoverable Robust Problem]\label{def:CombinatorialRecoverableRobustProblem}
    Let an SSP problem $\Pi = (\I, \U, \sol)$ be given.
    The combinatorial recoverable robust problem associated to $\Pi$ is denoted by \textsc{Comb. RR-$\Pi$} and defined as follows:
    The input is an instance $I \in \I$ together with a set $B \subseteq \U(I)$, an uncertainty parameter $\Gamma \in \N_0$, and a recoverability parameter $\kappa \in \N_0$.
    The question is whether
    \begin{align*}
        \exists S_1 \subseteq \U(I) : \forall B' \subseteq B \ \textit{with} \ |B'| \leq \Gamma : \exists S_2 \subseteq \U(I) : \qquad\qquad\qquad\qquad\qquad\qquad\qquad\\ 
        S_1, S_2 \in \sol(I), S_2 \cap B' = \emptyset \ \text{and} \ \dist(S_1, S_2) \leq \kappa.
    \end{align*}
\end{definition}

\section{An SSP Framework for Recoverable Robust Problems}

In this section, we prove our main result, i.e. we introduce and aprove a sufficient condition for nominal problems such that the corresponding recoverable robust problem is $\Sigma_3^p$-hard.
We first explain the rough idea.

We want to show that recoverable problems are $\Sigma^p_3$-hard, so we have to choose some $\Sigma^p_3$-hard problem from which we start our reduction.
The canonical $\Sigma^p_3$-complete problem is $\exists\forall\exists$-\textsc{Satisfiability} \cite{DBLP:journals/tcs/Stockmeyer76}.
Intuitively, this satisfiability problem can be understood as a game between Alice and Bob.
Alice first chooses a variable assignment on the variables of $X$, then Bob chooses an assignment of the variables $Y$, and then again Alice selects an assignment on the variables $Z$, where Alice wishes to satisfy formula $\varphi(X,Y,Z)$, and Bob wishes the opposite.
On the other hand, in our target problem RR-$\Pi$, any solution consists of two solutions of the nominal problem $\Pi$, the first stage solution $S_1$ and the second stage solution $S_2$.
The main challenge of our reduction is, that we have to model three variable sets $X, Y, Z$ of the formula $\exists X \forall Y \exists Z \varphi(X,Y,Z)$ in the order of quantification into the problem RR-$\Pi$.
Accordingly, both solutions $S_1, S_2$ need to include an assignment to all variables from $X, Y, Z$ that satisfy $\varphi(X,Y,Z)$ in agreement with the nominal \textsc{Sat} problem.
However, note that there is a difference in the solution structure of both problems, which we need to address.
If we model Alice's decision on the $X$-variables in $\exists\forall\exists$-\textsc{Sat} in first stage solution $S_1$ of RR-$\Pi$, we need to make sure that Alice is not able to reassign the the $X$-variables in the second stage solution $S_2$ of RR-$\Pi$.
Otherwise, Alice does not adhere to the order of quantification.

This problem can be circumvented by adding an additional property to the SSP reduction.
We demand that for a given set of literals $L_b$ the distance of two solutions $S_1$ and $S_2$ is small if and only if the partial solution of $S_1$ and $S_2$ restricted to $L_b$ is the same.
One can interpret this construction as a gadget that \emph{blows up} the literals of $L_b$ in comparison to all the other literals.
Accordingly, we call the literals of $L_b$ \emph{blow-up literals} and the corresponding reduction \emph{blow-up SSP reduction}.
Then, we can set $L_b = L_X$, where $L_X$ is the set of literals corresponding to variables of $X$.
Therefore, Alice is not able to change the assignment of the variable set $X$ from the first stage solution $S_1$ to the second stage solution $S_2$ and thus has to adhere to the order of quantification.
Then, it is also possible to reuse the injective correspondence function $f$ of the SSP reduction to set the blockable elements $B$ and to show the correctness of the reduction by using the elementary correspondence between the \textsc{Sat} solution and the solution of $\Pi$.
Formally, we define \emph{blow-up SSP reductions} as follows.

\begin{definition}[Blow-Up SSP Reduction]
\label{def:blow-up}
    Let $\dist$ be a distance measure.
    Let $\textsc{3Sat} = (\I,\U,\sol)$ with $\U = L = \{\ell_1, \ldots, \ell_n\} \cup \{\overline \ell_1, \ldots, \overline \ell_n\}$ and $\Pi' = (\I',\U',\sol')$ be two SSP problems.
    Then, \textsc{3Sat} is SSP blow-up reducible to $\Pi'$ with respect to $\dist(\cdot, \cdot)$ if for all sets $L_b \subseteq L$ fulfilling $\ell_i \in L_b \leftrightarrow \overline \ell_i \in L_b$ there exists an SSP reduction $(g, f_I)$ with the following property:
    There is a polynomial time computable blow-up factor $\beta_I \in \N$ corresponding to each instance $I$ such that for all solutions $S_1, S_2 \in \sol'(g(I))$ of $g(I)$:
    \[
        f(L_b) \cap S_1 = f(L_b) \cap S_2 \Leftrightarrow \dist_{\U'(g(I))}(S_1, S_2) \leq \beta_I.
    \]
\end{definition}

A concrete example of a blow-up reduction is given in \cref{sec:blow-up-ssp-reduction:reductions}.

\textit{Remark 1.}
Note that we define blow-up reductions to start at \textsc{3Sat}.
Because \textsc{3Sat} is one of the 'first' NP-complete problems, many reductions start at \textsc{3Sat} or have a short reduction chain from \textsc{3Sat}.
Thus, it is convenient to reuse these reductions or respectively reduction chains to construct blow-up reductions.
As we will show later in \Cref{sec:blow-up-ssp-reduction:reductions,sec:blow-up-preserving-ssp-reduction} this is not really a restriction.

\textit{Remark 2.}
Blow-up SSP reductions are not transitive.
(This is because there is no restriction on newly introduced elements in $\Pi'$ and how they behave corresponding to solutions in $\Pi'$.
Therefore, the distance between solutions in $\Pi'$ cannot be related to the distance between the corresponding solutions in $\Pi$.)
We will tackle this problem in \Cref{sec:blow-up-preserving-ssp-reduction}, in which we will show that we only need to find a blow-up reduction once at the beginning of a reduction chain and we can use so-called blow-up preserving reductions to append further problems to the reduction chain starting at \textsc{3Sat}.
This blow-up preserving reduction is typically much easier to find than a new blow-up reduction.

With everything in place, we show that blow-up SSP reductions enable us to show the $\Sigma^p_3$-hardness of recoverable robust problems as long as the nominal problem $\Pi$ is SSP-NP-hard.
In particular, our main theorem now states, that our newly introduced blow-up reductions indeed are a sufficient criterion for $\Sigma^p_3$-hardness.

\begin{theorem}\label{thm:rr-hardness}
    For all SSP-NP-hard problems $\Pi$ that are blow-up SSP reducible from \textsc{3Sat}, the combinatorial recoverable robust variant \textsc{Comb. RR-$\Pi$} is $\Sigma^p_3$-hard.
\end{theorem}
\begin{proof}

For the proof of the main theorem, we require some definitions. 
Let $X = \fromto{x_1}{x_n}$ be a set of binary variables and $L := \fromto{x_1}{x_n} \cup \fromto{\overline x_1}{\overline x_n} = X \cup \overline X$ be the set of corresponding literals.
An \emph{assignment} is a subset $A \subseteq X \cup \overline X$ such that $|A \cap \set{x_i, \overline x_i} | = 1$ for all $i=1,\dots,n$.
We say that the assignment $A$ assigns the value $1$ to variable $x_i$, if $x_i \in A$, and $A$ assigns $0$ to $x_i$, if $\overline x_i \in A$. We remark that this notation for an assignment is non-standard, but it turns out to be convenient in the context of our framework. 
A \textsc{Sat}-formula $\varphi$ is in 3CNF, if it is a conjunction of clauses, and every clause has exactly three literals. We denote by $\varphi(A) \in \set{0,1}$ the evaluation of the formula $\varphi$ under the assignment $A$.
In the following, we often consider the case, where the set of variables is partitioned into three disjoint sets $X, Y, Z$, and denote this case by writing $\varphi(X, Y, Z)$.

For the $\Sigma^p_3$-hardness proof, we require a known $\Sigma^p_3$-complete problem to reduce from. 
It turns out that instead of reducing from the classic problem  $\exists\forall\exists$-\textsc{Satisfiability}, it is more convenient to base our proof on \textsc{Robust Adjustable Sat} with budgeted uncertainty (in short, \textsc{R-Adj-Sat}), introduced by Goerigk, Lendl and Wulf \cite{DBLP:journals/dam/GoerigkLW24}, which we reformulate to adhere to the notation of this paper:

\begin{quote}
Problem $\textsc{R-Adj-Sat}$
\\
\textbf{Instance:}  A \textsc{Sat}-formula $\varphi(X,Y,Z)$ in 3CNF. 
A partition of the set of variables into three disjoint parts $X \cup Y \cup Z$ with $|X| = |Y| = |Z|$. An integer $\Gamma \geq 0$.  
\\
\textbf{Question:} Is there an assignment $A_X \subseteq X \cup \overline X$ such that for all subsets $Y' \subseteq Y$ of size $|Y'| \leq \Gamma$,
there exist assignments $A_Y \subseteq Y \cup \overline Y$ and $A_Z \subseteq Z \cup \overline Z$ which sets all variables in $Y'$ to 0, 
i.e. $A_Y \cap Y' = \emptyset$ and $\varphi(A_X, A_Y, A_Z) = 1$?
\end{quote}

The problem \textsc{R-Adj-Sat} is best understood as a game between Alice and Bob: 
Alice wants to satisfy the formula, Bob wants to achieve the opposite. 
Alice initially selects an assignment $A_X$ of the variables in $X$. 
Afterwards, Bob is allowed to select a blocker $Y' \subseteq Y$ of size at most $\Gamma$ and force all these variables to be \enquote{0}. 
Finally, Alice is allowed to assign values to all variables in $Y \cup Z$ that have not yet been assigned. In \cite{DBLP:journals/dam/GoerigkLW24} it is shown that \textsc{R-Adj-Sat} is $\Sigma^p_3$-complete. 
In order to provide additional insight for the reader, we quickly sketch the main argument in \cite{DBLP:journals/dam/GoerigkLW24}: 
The idea is to reduce from the classic problem $\exists\forall\exists$-\textsc{Satisfiability}. Let $\exists A_X \forall A_Y \exists A_Z \varphi(A_X, A_Y, A_Z)$ be an instance of this problem.
How do we transform it into a \textsc{R-Adj-Sat} problem?
It is intuitive, that in such a reduction, Alice's rule of choosing assignments $A_X$ and $A_Z$  should stay roughly the same. 
However, how do we express the choice of $A_Y$ in terms of a blocker $Y'$ that is forced to 0?
The idea is to split the variable $y_i \in Y$ into two new variables $y_i^t, y_i^f$ for every $i = 1,\dots,n$. 
We say that Bob plays honestly if for all $i = 1, \dots, n$ we have $|\set{y^t_i, y^f_i} \cap Y'| = 1$.
Such an honest choice of $Y'$ naturally encodes an assignment $A_Y$.
It turns out that, using the pigeonhole principle, and the budget constraint $|Y'| \leq \Gamma$, one can enrich the formula $\varphi$ with a \enquote{cheat-detection-gadget}.
This gadget can be used by Alice to win the game trivially if and only if Bob cheats.
Hence the modified game of $\textsc{R-Adj-Sat}$ is equivalent to the original instance of $\exists\forall\exists$-\textsc{Satisfiability}.

We prove $\Sigma^p_3$-hardness by providing a reduction from $\textsc{R-Adj-Sat}$ to \textsc{Comb. RR-$\Pi$}.
This reduction will crucially rely on the blow-up SSP reduction from \textsc{3Sat} to $\Pi$, which we assumed to exist.
More formally, let $\Pi = (\I_\Pi, \U_\Pi, \sol_\Pi)$ be an SSP problem, such that there is a blow-up reduction $(g, (f_I)_{I \in \I_\textsc{Sat}}, (\beta_I)_{I \in \I_\textsc{Sat}})$ from \textsc{3Sat} to $\Pi$. 
Recall that by the definition of a blow-up reduction this means the following:
Let the term $\I_\textsc{Sat}$ denote the set of all possible \textsc{3Sat} instances, i.e.\ the set of all 3CNF formulas. For each fixed \textsc{3Sat} instance $I_\textsc{Sat} \in \I_\textsc{Sat}$, its universe $\U(I_\textsc{Sat})$ is the set of positive and negative literals of variables appearing in that formula. 
By assumption, for every fixed \textsc{3Sat} instance $I_\textsc{Sat} \in \I_\textsc{Sat}$ and every set $L_b \subseteq \U(I_\textsc{Sat})$ (with $\ell \in L_b$ iff $\overline \ell \in L_b$) there exists an equivalent instance $I_\Pi = g(I_\textsc{Sat})$ of $\Pi$. 
Furthermore, associated to each fixed \textsc{3Sat} instance $I_\textsc{Sat} \in \I_\textsc{Sat}$ we have a function $f_{I_\textsc{Sat}}$, 
which describes in which way the universe of $I_\textsc{Sat}$ can be injectively embedded into the universe of the equivalent instance $I_\Pi$ such that the SSP property is true. 
Finally, there is a tight correspondence between solutions of $I_\Pi$ having distance at most $\beta_{I_\textsc{Sat}}$, and (the pre-image of) the solutions agreeing on the set $L_b$.
For brevity of notation, in the following we often write $\U(I_\Pi)$ instead of $\U_\Pi(I_\Pi)$ and $\sol(I_\Pi)$ instead of $\sol_\Pi(I_\Pi)$ to denote the solution set/universe associated to instance $I_\Pi$, if the correct subscript is clear from the context.

    We now turn our attention to the $\Sigma^p_3$-hardness proof.
    As explained above, the problem $\textsc{R-Adj-Sat}$ is $\Sigma^p_3$-complete.
    Let an $\textsc{R-Adj-Sat}$ instance $I_\text{RAS}$ be given.
    Our goal is to transform this instance into a new instance $I_\text{RR$\Pi$}$ of \textsc{Comb. RR-$\Pi$} such that
    \[
            I_\text{RAS} \ \text{is a Yes-instance of} \ \textsc{R-Adj-Sat} \Leftrightarrow I_\text{RR$\Pi$} \ \text{is a Yes-instance of Comb. RR-}\Pi.
    \]

    Clearly if we can achieve this goal we are done. By definition of the problem \textsc{R-Adj-Sat}, the instance $I_\text{RAS}$ consists out of the following parts:
    \[
        I_\text{RAS} = (\varphi(X, Y, Z), \Gamma),
    \]
    where $\varphi \in \I_\textsc{Sat}$ is a 3CNF-formula, whose variables are partitioned into three parts $X, Y, Z$ of equal size $n := |X| = |Y| = |Z|$, and $\Gamma \in \N_0$. 
    Let the corresponding literal sets $X,Y,Z$ be denoted by $L_X := X \cup \overline X$, $L_Y := Y \cup \overline Y,$ and $L_Z := Z \cup \overline Z$.
    Note that the universe associated to the \textsc{3Sat}-instance $\varphi$ is $L_X \cup L_Y \cup L_Z$. 
    Let us denote this fact by writing $\U(\varphi) = L_X \cup L_Y \cup L_Z$.

    Given the instance $I_\text{RAS}$ as input, we describe in the following how to construct the instance $I_\text{RR$\Pi$}$  of \textsc{Comb. RR-$\Pi$}, which consists out of the following parts:
    \[
        I_\text{RR$\Pi$} = (I_\Pi, B, \Gamma', \kappa).
    \]
    Here, $I_\Pi$ is an instance of the nominal problem $\Pi$, 
    $B \subseteq \U(I_\Pi)$ is the set of blockable elements, $\Gamma' \in \N_0$ is the uncertainty budget, and $\kappa \in \N_0$ is the recoverability parameter.
    
    Before we give a formal description of each of $I_\Pi, B, \Gamma', \kappa$, we explain the rough idea: 
    In particular, what is the right choice for the instance $I_\Pi$?
    The answer is provided by the blow-up reduction. 
    Note that this reduction can take in any 3CNF formula and a subset $L_b$ of the literals (with the property that $\ell \in L_b$ iff $\overline \ell \in L_b$ for all literals $\ell$) 
    and produce an instance $I$ of $\Pi$. Note that the set $L_b$ gets \enquote{blown up} by the reduction. We claim that $L_b = L_X$ is a good choice. 
    Indeed, consider the following formal definition of $I_\text{RR$\Pi$}$.

    \begin{description}
        \item[Description of the Instance $I_\text{RR$\Pi$}$.]
            Let $(g, (f_\varphi)_{\varphi \in \I_\textsc{Sat}}, (\beta_\varphi)_{\varphi \in \I_\textsc{Sat}})$ be the blow-up reduction from \textsc{3Sat} to $\Pi$.
            We let $I_\Pi = g(\varphi)$ be the instance produced by the blow-up reduction applied to the formula $\varphi$ and the literal set $L_b := L_X$.
            We remark that the following holds by the definition of the blow-up reduction:
            A blow up-reduction in particular is also an SSP reduction. 
            Hence the function
            $f_{\varphi} : \U(\varphi) \to \U(I_\Pi)$ maps the literals of $\varphi$ injectively to elements of the new universe $\U(I_\Pi)$. We can hence define the set of blockable elements
            \[ B: = f_\varphi(Y),\]
            where $Y \subseteq \U(\varphi)$ describes the positive literals in $L_Y$ (recall that $L_Y = Y \cup \overline Y$).

            Furthermore, note that by definition of a blow-up reduction, there exists some $\beta_\varphi \in \N$, computable in poly-time, with the property that for all $S_1, S_2 \in \sol(I_\Pi)$:
            \[
                \dist(S_1, S_2) \leq \beta_\varphi \Leftrightarrow S_1 \cap f_\varphi(L_X) = S_2 \cap f_\varphi(L_X).
            \]
            In other words, the new instance $I_\Pi$ has the property that two solutions of it have small distance if and only if they agree on $f_\varphi(L_X)$, 
            i.e.\ they agree on those universe elements of $\U(I_\Pi)$ that correspond to $L_X$.
            We finally define
            \[ \kappa := \beta_\varphi \text{ and } \Gamma' := \Gamma. \]
            This completes the description of the instance $I_\text{RR$\Pi$} = (I, B, \Gamma', \kappa)$.
        \item[Correctness]
        We start the correctness proof by showing
        \begin{align*}
            &I_\text{RAS} \ \text{is a Yes-instance} \Rightarrow I_\text{RR$\Pi$} = (I_\Pi, B, \Gamma, \kappa) \ \text{is a Yes-instance}.
        \end{align*}

        In this case, let $A_X \subseteq X \cup \overline X$ be an assignment of the variables $X$ with the property that 
        \[ \forall Y' \subseteq Y, |Y'| \leq \Gamma: \ \exists A_Y, A_Z : \ A_Y \cap Y' = \emptyset \text{ and } \varphi(A_X, A_Y, A_Z) = 1. \] 
        Such an $A_X$ exists by assumption that $I_\text{RAS}$ is a Yes-instance.
        Now, let $Y' \subseteq Y$ be an arbitrary subset of $Y$ with $|Y'| \leq \Gamma$.
        Then, it is possible to choose assignments $A_Y, A_Z$ of $Y, Z$ such that $A_Y \cap Y = \emptyset$ and $\varphi(A_X, A_Y, A_Z) = 1$.
        We consider such an assignment $A_1 := A_X \cup A_Y \cup A_Z$.
        Note that if we interpret $\varphi$ as \textsc{3Sat} instance, $A_1 \subseteq \U(\varphi)$ and even $A_1 \in \sol(\varphi)$, due to $\varphi(A_1) = 1$.
        Since the blow-up reduction is in particular an SSP reduction, we can make use of the central property of SSP reductions. 
        The property implies that the \textsc{3Sat} solution $A_1$ can be \enquote{lifted} to a solution of $\Pi$. 
        More precisely, there exists a solution $S_1 \in \sol(I_\Pi)$, such that 
        \[ S_1 \cap f_\varphi(\U(\varphi)) = f_\varphi(A_1). \]
        (Intuitively, the solution $S_1 \in \sol(I_\Pi)$ of the nominal problem instance $I_\Pi$ corresponds to the solution $A_1 \in \sol(\varphi)$ of the \textsc{3Sat} instance $\varphi$ when restricted to the injective embedded \enquote{sub-instance} of \textsc{3Sat} inside $\Pi$.)
        
        \quad We claim that $S_1$ is a solution for $I_\text{RR$\Pi$}$.
        To prove this claim we have to show that for all blockers $B' \subseteq B$ with $|B'| \leq \Gamma' = \Gamma$, there exists a solution $S_2 \in \sol(I_\Pi)$ with $S_2 \cap B' = \emptyset$ and $\dist(S_1, S_2) \leq \kappa$.
        Indeed, such a solution $S_2$ exists for all choices of blockers $B'$. 
        This can be seen by applying the following construction: Repeat exactly the same construction as for $S_1$, except that the set $Y'$ is not chosen arbitrarily.
        Instead, choose $Y'$ by letting $Y' := f^{-1}_\varphi(B')$.
        Note that by the definition of $B$, the set $Y'$ is well-defined, i.e. $Y' \subseteq Y$, $|Y'| \leq \Gamma$, and $f_\varphi(Y') = B'$.
        Repeating the same construction, assignment $A_X$ stays the same, but we obtain different assignments $A'_Y, A'_Z$ with $A'_Y \cap Y' = \emptyset$ and $\varphi(A_X, A'_Y, A'_Z) = 1$.
        We let $A_2 := A_X \cup A'_Y \cup A'_Z$.
        Again, by the SSP property there exists a solution $S_2 \in \sol(I_\Pi)$ such that $S_2 \cap f_\varphi(\U(\varphi)) = f_\varphi(A_2)$.
        Therefore we have
        \[ 
            \emptyset = A_2 \cap Y' = f_\varphi(A_2) \cap f_\varphi(Y') = (S_2 \cap f_\varphi(\U(\varphi))) \cap B' = S_2 \cap B'. 
        \]
        Note that $A_X$ is the same in both constructions.
        Therefore $S_1 \cap f_\varphi(L_X) = S_2 \cap f_\varphi(L_X)$, and hence by the definition of a blow-up reduction, we have $\dist(S_1, S_2) \leq \beta_\varphi = \kappa$.
        In total, we have shown that $S_1$ is a solution of $\Pi$ such that for every blocker $B' \subseteq B$, with $|B'| \leq \Gamma$, there exists a good solution $S_2$. 
        This shows that $I_\text{RR$\Pi$}$ is  yes-instance. 
        
        It remains to consider the reverse direction:
        \begin{align*}
             &I_\text{RR$\Pi$} = (I, B, \Gamma', \kappa) \text{ is a Yes-instance} \Rightarrow I_\text{RAS} \text{ is a Yes-instance}
        \end{align*}
        The strategy is very similar, with the difference that we use the SSP property and the blow-up property in the reverse direction. Consider some solution $S_1 \in \sol(I_\Pi)$ which satisfies
        \[
            \forall B' \subseteq B, |B'| \leq \Gamma': \exists S_2 \in \sol(I_\Pi), S_2 \cap B' = \emptyset: \dist(S_1, S_2) \leq \kappa.
        \]
        We have to show that there exists an assignment $A_X \subseteq X \cup \overline X$ such that for all subsets 
        $Y' \subseteq Y$ there are assignments $A_Y, A_Z$ with $A_Y \cap Y = \emptyset$ and $\varphi(A_X, A_Y, A_Z) = 1$.
        Indeed, to define $A_X$, consider solution $S_1$.
        By the SSP property, since $S_1 \in \sol(I_\Pi)$, 
        the set $f^{-1}_\varphi(S_1) \subseteq \U(\varphi)$ is a set of literals which satisfies $\varphi$. 
        We let $A_X$ be the restriction of that satisfying assignment to the variables in $X$. 
        More formally, we let $A_1 := f^{-1}_\varphi(S_1)$. 
        
        \quad We claim that this assignment $A_X$ proves that $I_\text{RAS}$ is a yes-instance.
        Indeed, let some set $Y' \subseteq Y$, with $|Y'| \leq \Gamma$, be given.
        We define the blocker to be $B' :=  f_\varphi(Y')$.
        Observe that $B' \subseteq f_\varphi(Y) = B$, and $|B'| \leq \Gamma$. 
        Therefore there exists $S_2 \in \sol(I)$ such that $S_2 \cap B' = \emptyset$ and $\dist(S_1, S_2) \leq \kappa$.
        We can again interpret the solution $S_2$ using the SSP property as an assignment $A_2 := f^{-1}_\varphi(S_2)$. By the SSP property, this assignment satisfies $\varphi$, that is $A_2 \in \sol(\varphi)$.
        Since $\dist(S_1, S_2) \leq \kappa = \beta_\varphi$, we have $S_1 \cap f_\varphi(Y) = S_2 \cap f_\varphi(Y)$ by the property of a blow-up reduction. This in turn implies $A_2 \cap (X \cup \overline X) = A_X$.
        Finally, we have $\emptyset = S_2 \cap B' = f^{-1}_\varphi(S_2) \cap f^{-1}_\varphi(B') = A_2 \cap Y'$. 
        This shows that $I_\text{RAS}$ is a yes-instance, and concludes the proof.

        \item[Polynomial Time.]
            The instance $I_\text{RR$\Pi$} = (I_\Pi, B, \Gamma', \kappa)$ can indeed be constructed in polynomial time. 
            The nominal instance $I_\Pi$ can be computed efficiently, since the map $g$ in the SSP reduction can be computed efficiently. 
            The set $B$ can be computed, since the map $f_\varphi$ in the SSP reduction can be computed efficiently. The number $\kappa = \beta_{\varphi}$ can be computed efficiently by definition of a blow-up reduction. The number $\Gamma$ remains the same.
    \end{description}
\end{proof}

We have shown that the combinatorial version is indeed $\Sigma^p_3$-hard as long as the nominal SSP problem $\Pi$ is SSP-NP-hard with a blow-up SSP reduction from \textsc{Satisifiability}.
This however does not directly imply that the linear optimization version \textsc{RR}-$\Pi$ is also $\Sigma^p_3$-hard.
We remedy this problem with the following short reduction that simulates the set of blockable elements in \textsc{Comb. RR}-$\Pi$ with the cost functions in \textsc{RR}-$\Pi$.

\begin{theorem}
	For all LOP problems $\Pi \in$ NP with the property that the SSP problem derived from it is blow-up SSP reducible from \textsc{3Sat}, the recoverable robust version RR-$\Pi$ is $\Sigma^p_3$-complete.
\end{theorem}
\begin{proof}
	Let $\Pi = (\I, \U, \F, d, t)$ be the LOP problem and \textsc{RR-$\Pi$} be the corresponding recoverable robust version.
	The containment of RR-$\Pi$ in $\Sigma^p_3$ follows from \Cref{thm:hddr-containment} by an analogous argument.

	For the hardness, let $\Pi' = (\I, \U, \sol)$ be the derived SSP problem from the LOP problem $\Pi$.
	Remember, by definition of derived SSP problems, the set of solutions is defined by $\sol(I) = \{F \in \F(I) : d^{(I)}(F) \leq t^{(I)}\}$.
    Now, let \textsc{Comb. RR}-$\Pi'$ the corresponding combinatorial recoverable robust version.
	By assumption that the SSP problem $\Pi'$ is SSP blow-up reducible from \textsc{Satisfiability} and \Cref{thm:rr-hardness}, \textsc{Comb. RR}-$\Pi'$ is $\Sigma^p_3$-hard.
    
    We want to reduce \textsc{Comb. RR}-$\Pi'$ to \textsc{RR}-$\Pi$.
    For this consider a \textsc{Comb. RR}-$\Pi'$-instance $(I', U', t', B', \Gamma', \kappa')$, which we transform to the \textsc{RR-$\Pi$}-instance $(I, U, c_1, \underline c, \overline c, t, \Gamma, \kappa)$.
    The instance and the universe stay the same, i.e. $I=I'$ and $U=U'$.
    We define the first stage cost function to be $c_1 = d^{(I)}$.
	We further have to model the blocker with the cost functions $\underline c$ and $\overline c$.
    For this, we define the cost functions such that the blockable elements can be correctly blocked by choosing the cost function $\underline c$ and $\overline c$ correspondingly:
	\[
		\underline c(u) = d^{(I)}(u), \ u \in \U(I)
	\]
	\[
		\overline c(u) = \begin{cases}
			\ d^{(I)}(u),& u \in \U(I) \setminus B \\
			\ 2t^{(I)}+1,& u \in B.
		\end{cases}			
	\]
	At last, we set $\Gamma' = \Gamma$, $\kappa' = \kappa$ and $t' = 2 t^{(I)}$.
	This completes the description of the instance of \textsc{RR}-$\Pi$.

	To prove the correctness of this reduction, we have to show that the following inequality holds if and only if $(I', U', t', B', \Gamma', \kappa')$ a Yes-instance:
	\begin{equation}\label{eq:comb-rr-reduction}
		\min_{\substack{S_1 \in \F(I)}} \max_{c_2 \in C_\Gamma} \min_{\substack{S_2 \in \F(I) \\ dist(S_1,S_2) \leq \kappa}} \ c_1(S_1) + c_2(S_2) \leq t'.
	\end{equation}

    For the first direction, assume that $(I', U', B', \Gamma', \kappa')$ is a Yes-instance.
    First, it holds that $|B'| \leq \Gamma'$.
    Furthermore, there are solutions $S_1, S_2 \subseteq U$ with $S_1, S_2 \in \sol(I)$ such that $\dist(S_1, S_2) \leq \kappa$ and $S_2 \cap B = \emptyset$.
    Since $\sol(I) = \{F \in \F(I) : d^{(I)}(F) \leq t^{(I)}\}$, we have that $S_1, S_2 \in \F(I)$ and $d^{(I)}(S_1) + d^{(I)}(S_2) \leq 2 t^{(I)}$.    
    At last, observe that the blocked elements $u \in B$ are assigned a cost of $2t^{(I)}+1$.
    Thus, these cannot be part of a second stage solution $S_2$ in \textsc{RR-$\Pi$} due to $t' = 2t^{(I)}$.
    Overall, all conditions on \Cref{eq:comb-rr-reduction} are fulfilled and $c_1(S_1) + c_2(S_2) \leq t'$ holds.
    
	For the second direction, assume that \Cref{eq:comb-rr-reduction} holds.
    That is, there is an $S_1 \subseteq \U(I)$ and an $S_2 \subseteq \U(I)$ with $S_1, S_2 \in \F(I)$ such that $\dist(S_1, S_2) \leq \kappa$.
    Furthermore, the set of cost functions $C_\Gamma$, defined by all cost functions $c_2$ such that $c_2(u) = \underline c(u) + \delta_u (\overline c(u) - \underline c(u)), \ \delta_u \in \{0,1\},$ and $\sum_{u \in \U(I)} \delta_u \leq \Gamma$, guarantees that at most $\Gamma$ elements have costs of $2t^{(I)} + 1$ while all others have $d^{(I)}(u)$.
    Since $t' = 2t^{(I)}$, $C_\Gamma$ simulates all possible blockers $B$ in the combinatorial version and it follows $S_2 \cap B = \emptyset$.
    Moreover, we have that $\sol(I) = \{F \in \F(I) : d^{(I)}(F) \leq t^{(I)}\}$, thus $S_1$ and $S_2$ are not only feasible solutions but it also holds that $S_1, S_2 \in \sol(I)$.
    All conditions of the solution pair $(S_1, S_2)$ are fulfilled to be a valid solution and we have a Yes-instance for the combinatorial version.
    Consequently, the reduction is correct.

    As last step, we have to prove that the reduction is polynomial time computable.
    This is indeed the case because the instance remains the same as well as the parameters $\Gamma$ and $\kappa$.
    All additional numbers to define the cost function $\overline c$ are $2t^{(I)} + 1$, which is linear in the input because $\Pi \in NP$.
    This also holds for the new threshold $t = 2t^{(I)}$.
    
    This concludes the proof and we have shown that \textsc{RR-$\Pi$} is indeed $\Sigma^p_3$-complete.
\end{proof}

\subsection{Blow-Up SSP Reductions for Various Problems}
\label{sec:blow-up-ssp-reduction:reductions}
With the theorems proven, we want to apply this framework to several well-known problems, in particular to 3-satisfiability, vertex cover, independent set, subset sum, directed Hamiltonian path, two directed disjoint path and Steiner tree.
In order to just convey the intuition how a blow-up SSP reduction works, we only present the reduction from satisfibility to vertex cover and defer the other reductions to the appendix.
Furthermore, we do not prove the correctness of the actual reduction as this is already presented in the original work, but we prove the correctness of the SSP and blow-up property. 
In \Cref{fig:blow-up-ssp-reduction-tree}, the tree of all presented reductions can be found beginning at \textsc{Satisfiability}.

For a blow-up SSP reduction, we need the following mappings:
The polynomial-time many-one reduction $g: \{0,1\}^* \rightarrow \{0,1\}^*$,
the injective SSP mapping $(f_I)_{I \in \I}: \U(I) \rightarrow \U'(g(I))$
and the polynomial-time computable blow-up factor $\beta_I \in \N$
such that $f(L_b) \cap S_1 = f(L_b) \cap S_2 \Leftrightarrow |S_1 \vartriangle S_2| \leq \beta_I$ for all solutions $S_1, S_2 \in \sol'(g(I))$ of $g(I)$.

\begin{figure}[!ht]
	\centering
	\resizebox{0.8\textwidth}{!}{
	\begin{tikzpicture}[scale=1]
		\node[text width=1cm,align=center](s) at (0, 1) {\textsc{Sat}};
		\node[text width=1cm,align=center](0) at (0, 0) {\textsc{3Sat}};
		\node[](1) at (-8, -1) {\textsc{VC}};
		\node[](2) at (-2.75, -1) {\textsc{IS}};
		\node[](3) at (0, -1) {\textsc{Subset Sum}};
		\node[](4) at (6.5, -0.85) {}; \node[text width=1cm,align=center](4-1) at (7, -1.19) {\textsc{Steiner Tree}};
		\node[](5) at (3.25, -1) {\textsc{DHamPath}};
		\node[](6) at (5.5, -1) {\textsc{2DDP}};
		\path[->] (s) edge (0);
		\path[->] (0) edge (1);
		\path[->] (0) edge (2);
		\path[->] (0) edge (3);
		\path[->] (0) edge (4);
		\path[->] (0) edge (5);
		\path[->] (0) edge (6);
	\end{tikzpicture}
	}
	\caption{The tree of SSP blow-up reductions for all considered problems.}
	\label{fig:blow-up-ssp-reduction-tree}
\end{figure}
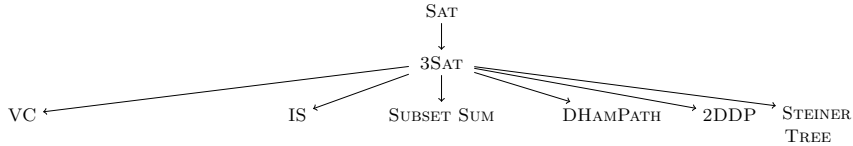

We start by defining the 3satisfiability problem in the SSP framework as presented in \cite{DBLP:journals/corr/abs-2311-10540}. Then, we define the vertex cover problem in the SSP framework.

\begin{samepage}
   	\begin{description}
        \item[]\textsc{3Satisfiability}\hfill\\
        \textbf{Instances:} Literal Set $L = \fromto{\ell_1}{\ell_n} \cup \fromto{\overline \ell_1}{\overline \ell_n}$, Clauses $C \subseteq L^3$.\\
        \textbf{Universe:} $L =: \U$.\\
        \textbf{Solution set:} The set of all sets $L' \subseteq \U$ such that for all $i \in \fromto{1}{n}$ we have $|L' \cap \set{\ell_i, \overline \ell_i}| = 1$, and such that $|L' \cap c| \geq 1$ for all $c \in C$.
   	\end{description}
\end{samepage}

\begin{samepage}
   	\begin{description}
       \item[]\textsc{Vertex Cover}\hfill\\
       \textbf{Instances:} Graph $G = (V, E)$, number $k \in \N$.\\
       \textbf{Universe:} Vertex set $V =: \U$.\\
       \textbf{Solution set:} The set of all vertex covers of size at most $k$.
   	\end{description}
\end{samepage}

Now, we can use a modification of the reduction by Garey and Johnson \cite{DBLP:books/fm/GareyJ79} from \textsc{Sat} to \textsc{Vertex Cover} as blow-up SSP reduction.
We first describe the original reduction and then argue why it is SSP.
At last, we show how to modify the reduction to a blow-up SSP reduction with a blow-up factor $\beta_I$.

Let $(L, C)$ be the \textsc{3Sat} instance of literals $L$ and clauses $C$.
The corresponding \textsc{Vertex Cover} instance $(G, k)$ with graph $G=(V,E)$ and integer $k$ is then constructed as follows.
The reduction maps each literal $\ell$ to a vertex $v_\ell$ and we denote the set of all these vertices $v_\ell$ by $W$.
Then the vertices $v_\ell$ and $v_{\overline \ell}$ of a literal $\ell$ and its negation $\overline \ell$ are connected by the edge $\{v_\ell, v_{\overline \ell}\}$.
Next, for all clauses $c \in C$ there is a 3-clique.
That is, for every literal $\ell$ in the clause $c$, there is a vertex $v^c_\ell$, which is connected to all other vertices $v^c_{\ell'}$ with $\ell' \in c$ and $\ell' \neq \ell$.
At last, for all literal vertices $v_\ell$ we add an edge $\{v_\ell, v^c_\ell\}$ to all clause vertices related to $\ell$.
The threshold of the vertex cover instance is then set to $k=|L|/2 + 2|C|$.
An example of the reduction can be found in \Cref{fig:reduction:3sat-vertex-cover}.

\tikzstyle{vertex}=[draw,circle,fill=black, minimum size=4pt,inner sep=0pt]
\tikzstyle{edge} = [draw,-]
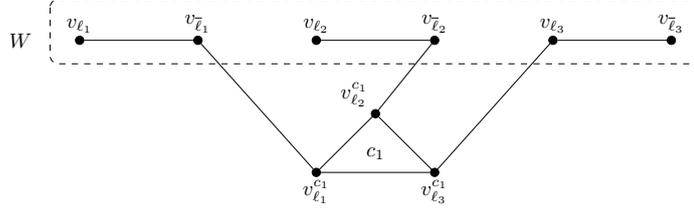
\begin{figure}[thpb]
\centering
\resizebox{0.67\textwidth}{!}{
\begin{tikzpicture}[scale=1,auto]

\node[vertex] (x1) at (0,0) {}; \node[above] at (x1) {$v_{\ell_1}$};
\node[vertex] (notx1) at (2,0) {}; \node[above] at (notx1) {$v_{\overline \ell_1}$};
\draw[edge] (x1) to (notx1);

\node[vertex] (x2) at (4,0) {}; \node[above] at (x2) {$v_{\ell_2}$};
\node[vertex] (notx2) at (6,0) {}; \node[above] at (notx2) {$v_{\overline \ell_2}$};
\draw[edge] (x2) to (notx2);

\node[vertex] (x3) at (8,0) {}; \node[above] at (x3) {$v_{\ell_3}$};
\node[vertex] (notx3) at (10,0) {}; \node[above] at (notx3) {$v_{\overline \ell_3}$};
\draw[edge] (x3) to (notx3);

\node[vertex] (c1) at (4,-2.25) {}; \node[below] at (c1) {$v^{c_1}_{\ell_1}$};
\node[vertex] (c2) at (5,-1.25) {}; \node[above left] at (c2) {$v^{c_1}_{\ell_2}$};
\node[vertex] (c3) at (6,-2.25) {}; \node[below] at (c3) {$v^{c_1}_{\ell_3}$};
\node[] at (5,-1.92) {$c_1$};
\draw[edge] (c1) to (c2) to (c3) to (c1);
\draw[edge] (notx1) to (c1);
\draw[edge] (notx2) to (c2);
\draw[edge] (x3) to (c3);

\node at ($(x1)+(-1,0)$) {$W$};
\draw[dashed,rounded corners] ($(x1)+(-.5,+.7)$) rectangle ($(notx3) + (.5,-.4)$);

\end{tikzpicture}
}
\caption{Classic reduction of \textsc{3Sat} to \textsc{Vertex Cover} for $\varphi = (\overline \ell_1 \lor \overline \ell_2 \lor \ell_3)$.}
\label{fig:reduction:3sat-vertex-cover}
\end{figure}

\begin{claim}
    The reduction from above from \textsc{3Sat} to \textsc{Vertex Cover} is an SSP reduction.
\end{claim}
\begin{claimproof}
    The reduction is an SSP reduction because each literal $\ell$ corresponds to exactly one vertex $v_\ell$ and the literal $\ell$ is in the satisfiability solution if and only if the corresponding vertex $v_\ell$ is in the vertex cover solution. 
    Thus, we are able to prove the following equality, which confirm that the reduction is SSP.
    
    \begin{align*}
        \{f(S) : S \subseteq L \ \text{s.t.} \ S \in \sol_{\textsc{3Sat}}\} &=
        \{\{f(\ell) : \ell \in S\} : S \subseteq L \ \text{s.t.} \ S \in \sol_{\textsc{3Sat}}\} \\
        &= \{\{v_\ell : \ell \in S\} : S \subseteq L \ \text{s.t.} \ S \in \sol_{\textsc{3Sat}}\} \\
        &= \{\{v_\ell : v_\ell \in S' \cap f(L)\} : S' \cap f(L) \subseteq W \ \text{s.t.} \ S' \in \sol_{\textsc{VC}}\} \\
        &= \{S' \cap f(L) : S' \in \sol_{\textsc{VC}}\}.
    \end{align*}

    The correctness proof of the reduction is presented in \cite{DBLP:books/fm/GareyJ79}.
\end{claimproof}

We have shown that the original reduction by Garey and Johnson is an SSP reduction.
For a modification to a blow-up SSP reduction, we introduce a blow-up gadget that we attach to each pair of blow-up literals $\ell, \overline \ell \in L_b$.
Introducing a blow-up gadget that is attached to the blow-up literals is a standard technique to construct blow-up SSP reductions and we use it for all of the following blow-up SSP reductions.
The idea is to build up equivalence classes $Q_{\ell}$ and $Q_{\overline \ell}$ of universe elements for both of the blow-up literals $\ell$ and $\overline \ell$ to which the gadget is connected to.
Such a class $Q_{\ell}$ has the property that $\ell$ is in the solution if and only if all elements of $Q_{\ell}$ are in the solution.
In other words, if one wants to switch from $\ell$ to $\overline \ell$, then additionally all elements from the equivalence class $Q_{\ell}$ have to be switched to $Q_{\overline \ell}$.
Furthermore, a blow-up gadget is variable in size without influencing the rest of the construction, i.e. for each instance we are able to find a large enough $\beta_I$ that locally guarantees that $\ell$ is not switchable to $\overline \ell$ while retaining a small distance between both solutions.

Given the distance measure, $\beta_I$ for every possible \textsc{3Sat} instance $I$ and blow-up literals $L_b$, we now describe the blow-up gadget.
The blow-up gadget for $\ell, \overline \ell \in L_b$ is a duplication of the literal vertices $v_\ell$ and $v_{\overline \ell}$ forming a complete bipartite graph $K_{\beta_I+1, \beta_I+1}$ between the vertex sets $V^{\beta_I}_{\ell} = \{v_\ell, v_{\ell^1}, \ldots, v_{\ell^{\beta_I}}\}$ and $V^{\beta_I}_{\overline \ell} = \{v_{\overline \ell}, v_{{\overline \ell}^1}, \ldots, v_{{\overline \ell}^{\beta_I}}\}$ as depicted in \Cref{fig:blow-up:3sat-vertex-cover}.
Then, we have the equivalence classes $Q_{\ell} = \{v_{\ell^1}, \ldots, v_{\ell^{\beta_I}}\}$ and $Q_{\overline \ell} = \{v_{\overline \ell^1}, \ldots, v_{\overline \ell^{\beta_I}}\}$.
Now, consider the graph $G$ as described by Garey and Johnson and modify it the following way:
For each pair $\ell, \overline \ell \in L_b$, we identify the two vertices $v_\ell, v_{\overline \ell}$ in $G$ and in the blow-up gadget for $\ell, \overline \ell \in L_b$ and merge them together.
This results in the graph
$$
    G'(L_b, \beta_I) =  ( V(G) \cup \bigcup_{\ell \in L_b} V^{\beta_I}_{\ell}, \ E(G) \cup \bigcup_{\ell \in L_b} \{\{a, b\} \mid a \in V^{\beta_I}_\ell, b \in V^{\beta_I}_{\overline \ell}\})
$$

The threshold for the vertex cover is then increased such that $\beta_I+1$ vertices (instead of only one) for every blow-up literal pair can be taken into the solution, i.e.
$$
    k'(\beta_I) = (\beta_I+1) \cdot |L_b|/2 + |L \setminus L_b|/2 + 2|C|.
$$

\tikzstyle{vertex}=[draw,circle,fill=black, minimum size=4pt,inner sep=0pt]
\tikzstyle{edge} = [draw,-]
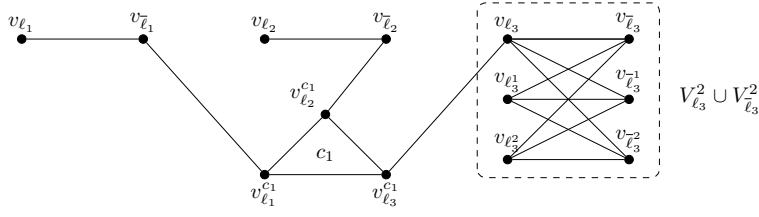
\begin{figure}
\centering
\resizebox{0.75\textwidth}{!}{
\begin{tikzpicture}[scale=1,auto]

\node[vertex] (x11) at (0,0) {}; \node[above] at (x11) {$v_{\ell_1}$};
\node[vertex] (notx11) at (2,0) {}; \node[above] at (notx11) {$v_{\overline \ell_1}$};
\draw[edge] (x11) to (notx11);

\node[vertex] (x21) at (4,0) {}; \node[above] at (x21) {$v_{\ell_2}$};
\node[vertex] (notx21) at (6,0) {}; \node[above] at (notx21) {$v_{\overline \ell_2}$};
\draw[edge] (x21) to (notx21);

\node[vertex] (x31) at (8,0) {}; \node[above] at (x31) {$v_{\ell_3}$};
\node[vertex] (notx31) at (10,0) {}; \node[above] at (notx31) {$v_{\overline \ell_3}$};
\draw[edge] (x31) to (notx31);

\node[vertex] (c1) at (4,-2.25) {}; \node[below] at (c1) {$v^{c_1}_{\ell_1}$};
\node[vertex] (c2) at (5,-1.25) {}; \node[above left] at (c2) {$v^{c_1}_{\ell_2}$};
\node[vertex] (c3) at (6,-2.25) {}; \node[below] at (c3) {$v^{c_1}_{\ell_3}$};
\node[] at (5,-1.92) {$c_1$};
\draw[edge] (c1) to (c2) to (c3) to (c1);
\draw[edge] (notx11) to (c1);
\draw[edge] (notx21) to (c2);
\draw[edge] (x31) to (c3);

\node[vertex] (x32) at ($(x31) + (0,-1)$) {}; \node[above] at (x32) {$v_{\ell^1_3}$};
\node[vertex] (notx32) at ($(x31) + (2,-1)$) {}; \node[above] at (notx32) {$v_{\overline \ell^1_3}$};
\node[vertex] (x33) at ($(x31) + (0,-2)$) {}; \node[above] at (x33) {$v_{\ell^2_3}$};
\node[vertex] (notx33) at ($(x31) + (2,-2)$) {}; \node[above] at (notx33) {$v_{\overline \ell^2_3}$};
\draw[edge] (x31) to (notx31);
\draw[edge] (x31) to (notx32);
\draw[edge] (x31) to (notx33);
\draw[edge] (x32) to (notx31);
\draw[edge] (x32) to (notx32);
\draw[edge] (x32) to (notx33);
\draw[edge] (x33) to (notx31);
\draw[edge] (x33) to (notx32);
\draw[edge] (x33) to (notx33);
\node[] (x11c1) at ($(x11) + (-0.5,-0.7)$) {};
\node[] (x11c2) at ($(x11) + (0,-0.7)$) {};
\node[] (x11c3) at ($(x11) + (0.5,-0.7)$) {};
\node[] (notx11c1) at ($(notx11) + (-0.5,-0.7)$) {};
\node[] (notx11c2) at ($(notx11) + (0,-0.7)$) {};
\node[] (notx11c3) at ($(notx11) + (0.5,-0.7)$) {};

\node at ($(notx32)+(1.5,0)$) {$V^2_{\ell_3} \cup V^2_{\overline \ell_3}$};
\draw[dashed,rounded corners] ($(x33)+(-.5,-.3)$) rectangle ($(notx31) + (.5,+.6)$);

\end{tikzpicture}
}
\caption{The graph $G'(L_b , \beta_I)$ with $L_b = \{\ell_3, \overline \ell_3\}$ and $\beta_I = 2$ of the instance $\varphi = (\overline \ell_1 \lor \overline \ell_2 \lor \ell_3)$. The blow-up gadget of the literals $\ell_3$ and $\overline \ell_3$ consiting of vetex sets $V^2_{\ell_3} \cup V^2_{\overline \ell_3}$ is outlined with dashed lines.}
\label{fig:blow-up:3sat-vertex-cover}
\end{figure}

\begin{claim}
    For the three choices of $\dist$ from \cref{subsec:distance-mesures} and $\beta_I := |V(G)|$, $G'(L_b, \beta_I)$ and $k'(\beta_I))$ describe a blow-up SSP-reduction from \textsc{3Sat} to \textsc{Vertex Cover}.
\end{claim}
\begin{claimproof}
    We show that $(g, (f_I)_{I \in \I}, (\beta_I)_{I \in \I})$ with 
    $$
        g(I) = (G'(L_b, {\beta_I}), k'(\beta_I)), \ f(\ell) = v_\ell, \text{ and } \beta_I = |V(G)|
    $$
    is a blow-up SSP reduction from \textsc{3Sat} to $\Pi$.
    We start by showing that the reduction is still SSP.

    For this, we first make the observation that for a bipartite graph $K_{\beta_I+1, \beta_I+1}$ between the vertex sets $V^{\beta_I}_\ell$ and $V^{\beta_I}_{\overline \ell}$, there are exactly two vertex covers of size $\beta_I+1$:
    either one takes all $v_{\ell^i}$ together with $v_\ell$ or all $v_{\overline \ell^i}$ together with $v_{\overline \ell}$ into the vertex cover.
    Otherwise, there is at least one edge $\{v,w\}$ with $v \in V^{\beta_I}_\ell$ and $w \in V^{\beta_I}_{\overline \ell}$ not covered.

    With this observation, we can follow that every vertex cover solution still consists of
    \begin{enumerate}
        \item exactly one of the original vertices $v_\ell$ and $v_{\overline \ell}$ for $\ell, \overline \ell \in L$, and
        \item the $\beta_I$ additional vertices $\{v_{\ell^i} : 1 \leq i \leq \beta_I\}$ of the blow-up gadget if and only if $v_\ell \in L_b$ is in the vertex cover solution, and
        \item exactly two of the three vertices $v^{c_j}_{\ell_1}, v^{c_j}_{\ell_2}, v^{c_j}_{\ell_3}$ for $c_j \in C$.
    \end{enumerate}
    Thus, the original injective SSP mapping $f(\ell) = v_\ell$ is still a valid SSP mapping for this reduction.

    To show that the blow-up property holds, we have to show that for $\beta_I := |V(G)|$ and for all solutions of $g(I)$, $S_1, S_2 \in \sol'(g(I))$,
    \begin{equation}\label{eq:claim_blow-up}
        f(L_b) \cap S_1 = f(L_b) \cap S_2 \Leftrightarrow \dist(S_1,S_2) \leq \beta_I.
    \end{equation}

    We begin with analyzing the maximum distances of two solutions $S_1, S_2 \in \sol'(g(I))$ for the three distance measures if $S_1$ and $S_2$ agree on each blow-up gadget.
    To reach the maximum distance between two solutions $S_1$ and $S_2$ for the three distance measures,
    $S_1$ includes all the $v_\ell$ for all $\ell \in L \setminus L_b$, while $S_2$ includes all the opposite $v_{\overline \ell}$.
    Additionally, $S_1$ includes a different pair $v^{c_j}_{\ell_1}, v^{c_j}_{\ell_2}$ for all $c_j \in C$ in comparison to $S_2$.
    Thus, all $v_{\overline \ell}$ of $\overline \ell \in L \setminus L_b$ and $v^{c_j}_{\ell_3}$ for all $c_j \in C$ have to be added (while deleting $v_\ell$ and one of $v^{c_j}_{\ell_1}, v^{c_j}_{\ell_2}$).
    In conclusion, $\beta_I = |V(G)|$ is an upper bound for the maximal possible \emph{$\kappa$-addition} and \emph{$\kappa$-deletion} distance ($|C| + |L \setminus L_b|/2$) as well as for the maximal possible \emph{Hamming distance} (here one has to count for all deletions \emph{and} additions such that $\beta_I = |V(G)| \geq 2|C| + |L \setminus L_b| \geq |S_1 \vartriangle S_2|$ is also sufficient).
    It follows that $f(L_b) \cap S_1 = f(L_b) \cap S_2 \Rightarrow \dist(S_1, S_2) \leq \beta_I$.

    We conclude with the situation that $S_1$ and $S_2$ do not agree on each blow-up gadget.
    Then corresponding to the analysis above for \emph{$\kappa$-addition} and \emph{$\kappa$-deletion}, we reach at least a distance of $|V(G)|+1 > \beta_I$ and for the Hamming distance we reach at least $2|V(G)|+2 > \beta_I$.
    It follows that $f(L_b) \cap S_1 \neq f(L_b) \cap S_2 \Rightarrow \dist(S_1, S_2) > \beta_I$.

    Clearly $\beta_I$ and thus $G'(L_b, \beta_I)$ and $k'(\beta_I)$ are polynomial time computable.
    This completes the correctness proof of the reduction.
\end{claimproof}

In conclusion, we have shown that the existing reduction from satisfiability to vertex cover enhanced with the presented blow-up gadget is a blow-up SSP reduction.
It follows that the combinatorial as well as the linear optimization version of recoverable robust vertex cover is  $\Sigma^p_3$-complete.

\section{The Issue of Transitivity: Preserving the Blow-up Gadget}\label{sec:blow-up-preserving-ssp-reduction}

Since blow-up SSP reductions require more structure, we unfortunately lose transitivity in comparison to normal SSP reductions, which are transitive.
However, we would like to reuse an existing blow-up SSP reduction such that we do not need to start every reduction at \textsc{3Satisfiability} and additionally construct blow-up gadgets.
For this, we introduce blow-up preserving SSP reductions, which are transitive and preserve the structure of the blow-up gadgets.
The idea is that we only need to show that there is a blow-up SSP reduction for the first reduction in the reduction chain and then the reduction chain can be prolonged by adding further problems, between which there are blow-up preserving SSP reductions.

The idea of a blow-up preserving reduction between problem $\Pi$ and problem $\Pi'$ is that we partition the universe into three sets, the set of elements which originate in the problem $\Pi$, a set of elements $U_{\textit{off}}$, which are never of a solution of the instance of $\Pi'$, and a set of elements $U_{\textit{on}}$, which are part of every solution of the instance of $\Pi'$.
Therefore, every solution $S$ in $\Pi$ has a correspondent solution $S' = f(S) \cup U_{\textit{on}}$ by applying the injective function $f$.
Since the distance measures that we consider are invariant on injective functions and union, the distance between the solutions in problem $\Pi$ is the same as in problem $\Pi'$.
Correspondingly, the blow-up SSP reduction from \textsc{3Satisfiability} to $\Pi$ can be extended to $\Pi'$.

From a different point of view, a blow-up preserving SSP reduction can be understood as an SSP reduction, which \enquote{adds} only two kind of new elements to the universe:
Those that are trivially contained in every (optimal) solution, and those that are never contained in an (optimal) solution.
It is intuitively easy to understand that compared to the old instance, the newly \enquote{added} elements can not influence the term $\dist(S_1,S_2)$.
Hence $\Sigma^p_3$-hardness of recoverable problems is maintained.

\begin{definition}[Blow-Up Preserving SSP Reduction]
\label{def:blow-up-preserving}
    Let $\Pi = (\I,\U,\sol)$ and $\Pi' = (\I',\U',\sol')$ be two SSP problems.
    Then, there is a blow-up preserving SSP reduction from $\Pi$ to $\Pi'$ if there exists an SSP reduction $(g, (f_I)_{I \in \I})$ such that for all instances $I \in \I$ the following holds:\\
    For all elements $u' \in \U'(g(I))$, either
    \begin{itemize}
        \item $u' \in f(\U(I))$ or
        \item for all $S' \in \sol': u' \notin S'$, i.e. $u' \in U_{\textit{off}}$, or
        \item for all $S' \in \sol': u' \in S'$, i.e. $u' \in U_{\textit{on}}$.
    \end{itemize}
\end{definition}

\textit{Remark.}
This definition is more restrictive than it needs to be because we do not need to force that the distance stays exactly the same.
It is also possible to increase the distance while applying the reduction in a controlled manner.
More precisely, we can allow $\gamma$ of the elements $u'$ to be in $U' \setminus (f(\U) \cup U_{\textit{off}} \cup U_{\textit{on}})$ and we relax the backward direction by $f(L_b) \cap S_1 = f(L_b) \cap S_2 \Rightarrow |S_1 \vartriangle S_2| > \beta + \gamma.$
This is a more general and more complicated definition, nevertheless, it is not necessary for the problems that are presented in this paper.
For simplicity, we use the stricter variant for the rest of the paper.

\begin{theorem}
    Let $\Pi$ be an SSP-NP-complete problem for which a blow-up SSP reduction from \textsc{3Sat} exists.
    Then every problem $\Pi'$, which is blow-up preserving SSP reducible from $\Pi$, is blow-up SSP reducible from \textsc{3Sat}.
\end{theorem}
\begin{proof}
    We present a blow-up reduction from \textsc{3Sat} to \textsc{$\Pi'$}.
    
    For this let \textsc{3Sat} $= (\I_\textsc{3Sat}, \U_\textsc{3Sat}, \sol_\textsc{3Sat})$, $\Pi = (\I_\Pi,\U_\Pi,\sol_\Pi)$ and $\Pi' = (\I_{\Pi'},\U_{\Pi'},\sol_{\Pi'})$.
    Moreover, we are given a \emph{blow-up SSP reduction} $(g, (f_{I})_{I \in \I_\textsc{3Sat}}, (\beta_I)_{I \in \I_\textsc{3Sat}})$ from \textsc{3Sat} to \textsc{$\Pi$} and a \emph{blow-up preserving SSP reduction} $(g', (f'_I)_{I \in \I_\Pi})$ from \textsc{$\Pi$} to \textsc{$\Pi'$}.

    We concatenate the blow-up SSP reduction $(g, (f_I)_{I \in \I_\textsc{3Sat}}, (\beta_I)_{I \in \I_\textsc{3Sat}})$ and the blow-up preserving SSP reduction $(g', (f'_I)_{I \in \I})$ to obtain a new blow-up SSP reduction \linebreak $(g' \circ g, (f'_{g(I)} \circ f_I)_{I \in \I_\textsc{3Sat}}, (\beta_I)_{I \in \I_\textsc{3Sat}})$ from \textsc{3Sat} to $\Pi'$.
    The resulting relation between the three problem universes is depicted in \Cref{fig:blow-up-preserving-reduction}.
    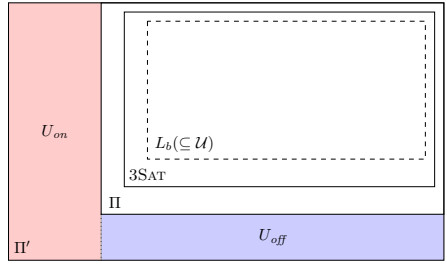
\begin{figure}[!ht]
        \centering
        \resizebox{0.42\textwidth}{!}{
            \begin{tikzpicture}
                \draw[dashed] ($(0,0)$) rectangle ($(6,-3)$);
                \draw[] ($(-0.5,0.2)$) rectangle ($(6.2,-3.6)$);
                \draw[] ($(-1,0.4)$) rectangle ($(6.4,-4.2)$);
                \draw[] ($(-3,0.4)$) rectangle ($(6.4,-5.2)$);
                \draw[fill=blue, opacity=0.2] (-1,-4.2) rectangle (6.4,-5.2);
                \draw[fill=red, opacity=0.2] (-3,0.4) rectangle (-1,-5.2);
                \draw[dotted] (-1,-4.2) edge (-1,-5.2);
                
        		\node[] () at (0.8,-2.67) {$L_b (\subseteq \mathcal{U})$};
        		\node[] () at (0,-3.33) {\textsc{3Sat}};
        		\node[] () at (-0.7,-3.95) {$\Pi$};
        		\node[] () at (-2.7,-4.9) {$\Pi'$};
        		\node[] () at (2.7,-4.7) {$U_{\textit{off}}$};
        		\node[] () at (-2,-2.4) {$U_{\textit{on}}$};
            \end{tikzpicture}
        }
        \caption{The relation between \textsc{3Sat}, $\Pi$ and $\Pi'$ in a blow-up SSP reduction from \textsc{3Sat} to $\Pi$ and a blow-up preserving SSP reduction from $\Pi$ and $\Pi'$.
        The blow-up preserving SSP reduction maps the universe of $\Pi$ into the universe elements of $\Pi'$.
        Because the original \textsc{3Sat} universe is part of the universe of $\Pi$, the blow-up literals $L_b$ are also mapped via the functions $f$ and $f'$ into the universe of $\Pi'$.
        Furthermore, all additional elements in $\Pi'$ are either part of $U_{\textit{on}}$, which are always part of a solution in $\Pi'$, or $U_{\textit{off}}$, which are never part of the solution in $\Pi'$.}
        \label{fig:blow-up-preserving-reduction}
    \end{figure}

    Blow-up SSP reductions and blow-up preserving SSP reductions are both SSP reductions.
    Furthermore, SSP reductions are transitive [\cite{DBLP:journals/corr/abs-2311-10540}, Lemma 5] and thus the concatenation of both reductions is still an SSP reduction.

    It remains to show that the blow-up property holds.
    We show this by using the same polynomial computable $\beta_I$ of the blow-up SSP reduction $(g, (f_I)_{I \in \I_\textsc{3Sat}}, (\beta_I)_{I \in \I_\textsc{3Sat}})$ from \textsc{3Sat} to $\Pi$.
    In particular, we have to show that the statement
    \begin{equation}\label{eq:blow-up-preserving:pi}
        \forall S_1, S_2 \in \sol_\Pi: \ f(L_b) \cap S_1 = f(L_b) \cap S_2 \Leftrightarrow \dist(S_1,S_2) \leq \beta_I.
    \end{equation}
    holds if and only if the following statement holds:
    \begin{equation}\label{eq:blow-up-preserving:piprime}
        \forall S'_1, S'_2 \in \sol_{\Pi'}: \ f' \circ f(L_b) \cap S'_1 = f' \circ f(L_b) \cap S'_2 \Leftrightarrow \dist(S'_1, S'_2) \leq \beta_I.
    \end{equation}
    We first remark that $\{S' : S' \in \sol_{\Pi'}\} = \{f'(S) \cup U_{\textit{on}} : S \in \sol_{\Pi}\}$.
    Thus, there is a bijection between solutions sets $S \in \sol_{\Pi}$ of problem $\Pi$ and the solution sets $S' \in \sol_{\Pi'}$ of problem $\Pi'$ defined by $h: \sol_{\Pi} \rightarrow \sol_{\Pi'}, S \mapsto f'(S) \cup U_{\textit{on}}$.
    With the help this bijection, we can show that the distance between the solution pair $(S_1, S_2)$ of problem $\Pi$ and the solution pair $(S'_1, S'_2)$ of problem $\Pi'$ stays the same.
    \begin{align}
        \text{dist}(S'_1, S'_2) &= \text{dist}(h(S_1), h(S_2)) \label{eq:blow-up-preserving_1}\\
        &= \text{dist}(f'(S_1) \cup U_{\textit{on}}, \ f'(S_2) \cup U_{\textit{on}}) \label{eq:blow-up-preserving_2}\\
        &= \text{dist}(f'(S_1), f'(S_2)) \label{eq:blow-up-preserving_3}\\
        &= \text{dist}(S_1, S_2) \label{eq:blow-up-preserving_4}
    \end{align}
    For \Cref{eq:blow-up-preserving_3}, we make use of the invariance of union of the distance measure.
    Furthermore \Cref{eq:blow-up-preserving_4} holds, because the function $f'$ is injective since the blow-up preserving SSP reduction is also an SSP reduction.
    Consequently, using the same $\beta_I$ guarantees that the equivalence of \Cref{eq:blow-up-preserving:pi} and \Cref{eq:blow-up-preserving:piprime} holds.

    The transformation is polynomial time computable because the mappings $g$ and $g'$ are polynomial time computable as well as the $\beta_I$ as given the in the blow-up SSP reduction $(g, (f_I)_{I \in \I_\textsc{3Sat}}, (\beta_I)_{I \in \I_\textsc{3Sat}})$.
    This concludes the proof of the theorem.
\end{proof}

\begin{corollary}
    Let $\Pi$ be an SSP-NP-hard problem for which a blow-up SSP reduction from \textsc{3Satisfiability} exists.
    Then for all SSP-NP-hard problems $\Pi'$ that are blow-up preserving SSP reducible from $\Pi$, the recoverable robust variant \textsc{Comb. RR-$\Pi'$} is $\Sigma^p_3$-hard.
\end{corollary}

We have shown that blow-up preserving SSP reductions are a possible tool to build up a reduction chain that begins at a blow-up SSP reduction.
Next, we prove that those chains are extendable by adding further blow-up preserving SSP reductions by showing the transitivity of these reductions.

\begin{lemma}
    Blow-up preserving SSP reductions are transitive.
\end{lemma}
\begin{proof}
    Let $\Pi^1$, $\Pi^2$ and $\Pi^3$ be three SSP problems such that $\Pi^1$ is blow-up preserving SSP reducible to $\Pi^2$ by $(g, f)$ and $\Pi^2$ is blow-up preserving SSP reducible to $\Pi^3$ by $(g', f')$.
    We aim to show that $\Pi^1$ is blow-up preserving SSP reducible to $\Pi^3$.

    Consider an instance $(I^1, U^1)$ of $\Pi^1$, which is mapped by the blow-up preserving SSP reduction $(I^2, U^2)$.
    Then,
    \[
        I^2 = g(I^1) \ \text{and} \ U^2 = f(U^1) \cup U^2_{\textit{on}} \cup U^2_{\textit{off}}.
    \]
    We can then use the second reduction $(g', f')$ to derive
    \[
        I^3 = g'(g(I^1)) \ \text{and} \ U^3 = f'(f(U^1)) \cup f'(U^2_{\textit{on}}) \cup f'(U^2_{\textit{off}}) \cup U^3_{\textit{on}} \cup U^3_{\textit{off}}.
    \]
    Then, we can reassign $U_{\textit{on}} = f'(U^2_{\textit{on}}) \cup U^3_{\textit{on}}$ and $U_{\textit{off}} = f'(U^2_{\textit{off}}) \cup U^3_{\textit{off}}$ as follows to derive
    \begin{equation}\label{eq:blow-up-preserving-transitive}
         U^3 = f'(f(U^1)) \cup U_{\textit{on}} \cup U_{\textit{off}}.
    \end{equation}
    Because $g$ and $g'$ are transitive as well as $f' \circ f$ is injective (since $f$ and $f'$ are injective), we have an SSP reduction, which is also blow-up preserving because of \Cref{eq:blow-up-preserving-transitive}.
\end{proof}

\subsection{Blow-up Gadgets for various Problems}\label{sec:blow-up-preserving-ssp-reduction:reductions}

In this subsection, we apply the developed framework of blow-up-preserving SSP reductions.
In \Cref{sec:blow-up-ssp-reduction:reductions}, we have seen that vertex cover, 3Sat, independent set, subset sum, directed Hamiltonian path, two directed disjoint path and Steiner tree are blow-up SSP reducible from satisfiability.
Next, we show that one can find blow-up preserving reductions starting from some of these problems to show $\Sigma^p_3$-hardness of further recoverable robust problems.
Again, we defer most of the reductions to the appendix.
However to convey the intuition, how a blow-up preserving SSP reduction works, we present a reduction from vertex cover to dominating set.
The complete tree of reductions can be found in \Cref{fig:blow-up-preserving-ssp-reduction-tree}.
Correspondingly, we can state that the recoverable robust problem based on the nominal problems in the reduction tree are $\Sigma^p_3$-complete.

\begin{theorem}
    The recoverable robust versions of the following nominal problems are $\Sigma^p_3$-complete:
    \textsc{Satisfiability},
    \textsc{3-Satisfiability},
    \textsc{Vertex Cover}, \textsc{Dominating Set}, \textsc{Set Cover}, \textsc{Hitting Set}, \textsc{Feedback Vertex Set}, \textsc{Feedback Arc Set}, \textsc{Uncapacitated Facility Location}, \textsc{p-Center}, \textsc{p-Median},
    \textsc{Independent Set}, \textsc{Clique},
    \textsc{Subset Sum}, \textsc{Knapsack}, \textsc{Partition}, \textsc{Scheduling},
    \textsc{Directed Hamiltonian Path}, \textsc{Directed Hamiltonian Cycle}, \textsc{Undirected Hamiltonian Cycle}, \textsc{Traveling Salesman Problem},
    \textsc{Two Directed Vertex Disjoint Path}, \textsc{$k$-Vertex Directed Disjoint Path},
    \textsc{Steiner Tree}
\end{theorem}

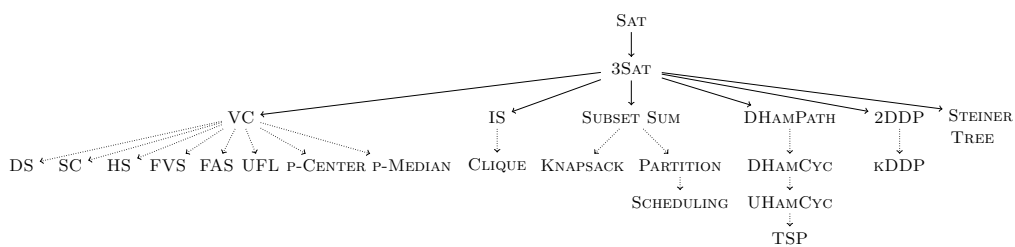
\begin{figure}[!ht]
	\centering
	\resizebox{0.95\textwidth}{!}{
	\begin{tikzpicture}[scale=1]
		\node[text width=1cm,align=center](s) at (0, 1) {\textsc{Sat}};
		\node[text width=1cm,align=center](0) at (0, 0) {\textsc{3Sat}};
		\node[](1) at (-8, -1) {\textsc{VC}};
		\node[](2) at (-2.75, -1) {\textsc{IS}};
		\node[](3) at (0, -1) {\textsc{Subset Sum}};
		\node[](4) at (6.5, -0.85) {}; \node[text width=1cm,align=center](4-1) at (7, -1.19) {\textsc{Steiner Tree}};
		\node[](5) at (3.25, -1) {\textsc{DHamPath}};
		\node[](6) at (5.5, -1) {\textsc{2DDP}};
		\node[](11) at (-12.5, -2) {\textsc{DS}};
		\node[](12) at (-11.5, -2) {\textsc{SC}};
		\node[](13) at (-10.5, -2) {\textsc{HS}};
		\node[](14) at (-9.5, -2) {\textsc{FVS}};
		\node[](15) at (-8.5, -2) {\textsc{FAS}};
		\node[](10) at (-7.6, -2) {\textsc{UFL}};
		\node[](16) at (-6.25, -2) {\textsc{p-Center}};
		\node[](17) at (-4.5, -2) {\textsc{p-Median}};
		\node[](21) at (-2.75, -2) {\textsc{Clique}};
		\node[](31) at (-1, -2) {\textsc{Knapsack}};
		\node[](32) at (1, -2) {\textsc{Partition}};
		\node[](51) at (3.25, -2) {\textsc{DHamCyc}};
		\node[](61) at (5.5, -2) {\textsc{kDDP}};
		\node[](321) at (1, -2.75) {\textsc{Scheduling}};
		\node[](511) at (3.25, -2.75) {\textsc{UHamCyc}};
        \node[](5111) at (3.25, -3.5) {\textsc{TSP}};
		\path[->] (s) edge (0);
		\path[->] (0) edge (1);
		\path[->] (0) edge (2);
		\path[->] (0) edge (3);
		\path[->] (0) edge (4);
		\path[->] (0) edge (5);
		\path[->] (0) edge (6);
		\path[->, densely dotted] (1) edge (10);
		\path[->, densely dotted] (1) edge (11);
		\path[->, densely dotted] (1) edge (12);
		\path[->, densely dotted] (1) edge (13);
		\path[->, densely dotted] (1) edge (14);
		\path[->, densely dotted] (1) edge (15);
		\path[->, densely dotted] (1) edge (16);
		\path[->, densely dotted] (1) edge (17);
		\path[->, densely dotted] (2) edge (21);
		\path[->, densely dotted] (3) edge (31);
		\path[->, densely dotted] (3) edge (32);
		\path[->, densely dotted] (32) edge (321);
		\path[->, densely dotted] (5) edge (51);
		\path[->, densely dotted] (51) edge (511);
		\path[->, densely dotted] (511) edge (5111);
		\path[->, densely dotted] (6) edge (61);
	\end{tikzpicture}
	}
	\caption{The tree of SSP reductions for all considered problems. While solid edges indicate that there is a blow-up SSP reduction between the problems, dotted edges represent a blow-up preserving SSP reduction between the problems.}
	\label{fig:blow-up-preserving-ssp-reduction-tree}
\end{figure}

We define dominating set according to the SSP framework as follows.

\begin{samepage}
	\begin{description}
        \item[]\textsc{Dominating Set}\hfill\\
        \textbf{Instances:} Graph $G = (V, E)$, number $k \in \N$.\\
        \textbf{Universe:} Vertex set $V =: \U$.\\
        \textbf{Solution set:} The set of all dominating sets of size at most $k$.
	\end{description}
\end{samepage}

For a reduction from \textsc{Vertex Cover} to \textsc{Dominating Set}, we use a modification of a folklore reduction.
The vertex cover instance, consisting of a graph $G = (V,E)$ and an integer $k$, is transformed to a dominating instance, consisting of a graph $G' = (V', E')$ and an integer $k'$.\footnote{For the sake of conciseness, we assume w.l.o.g. that the vertex cover instance is a connected graph. For a full proof for arbitrary instances we refer to \cite{DBLP:journals/corr/abs-2311-10540}.}
Every vertex $v \in V$ is mapped to a vertex $v' \in V'$, we denote the set of all these vertices $v'$ by $W$.
Every edge $\{v, w\} \in E$ is mapped to an edge $\{v', w'\} \in E'$.
Furthermore for every edge $\{v, w\} \in E$, there are $|V|+1$ additional vertices\footnote{The original reduction uses only one additional vertex for each edge. However, this reduction is not an SSP reduction.} $vw_i, 1 \leq i \leq |V|+1$.
These vertices are connected to both $v'$ and $w'$, i.e. $\{\{v', vw_i\}, \{w', vw_i\} : 1 \leq i \leq |V|+1\} \subseteq E'$.
At last, we set $k' = k$.
The transformation of one edge is depicted in \Cref{fig:reduction:vertex-cover-dominating-set}.

\tikzstyle{vertex}=[draw,circle,fill=black, minimum size=4pt,inner sep=0pt]
\tikzstyle{edge} = [draw,-]
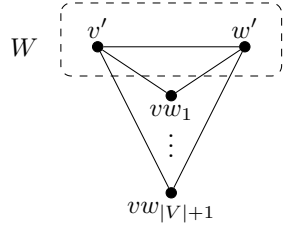
\begin{figure}[thpb]
\centering
\resizebox{0.28\textwidth}{!}{
\begin{tikzpicture}[scale=1,auto]

\node[vertex] (v) at (0,0) {}; \node[above] at (v) {$v'$};
\node[vertex] (w) at (2,0) {}; \node[above] at (w) {$w'$};
\draw[edge] (v) to (w);

\node[vertex] (vw1) at (1,-0.67) {}; \node[below] at (vw1) {$vw_1$};
\draw[edge] (v) to (vw1);
\draw[edge] (vw1) to (w);
\node[vertex] (vwV+1) at (1,-2) {}; \node[below] at (vwV+1) {$vw_{|V|+1}$};
\draw[edge] (v) to (vwV+1);
\draw[edge] (vwV+1) to (w);
\node[] () at (1,-1.25) {$\vdots$};

\node at ($(v)+(-1,0)$) {$W$};
\draw[dashed,rounded corners] ($(v)+(-.5,+.6)$) rectangle ($(w) + (.5,-.4)$);

\end{tikzpicture}
}
\caption{Modified reduction of \textsc{Vertex Cover} to \textsc{Dominating Set}.
This is the transformation for two vertices $v, w \in V$ connected by one edge $\{v,w\} \in E$.}
\label{fig:reduction:vertex-cover-dominating-set}
\end{figure}

\begin{claim}
    The reduction from above from \textsc{Vertex Cover} to \textsc{Dominating Set} is an SSP reduction.
\end{claim}
\begin{claimproof}
    This reduction is an SSP reduction, because every vertex $v \in V$ is mapped to a corresponding vertex $v' \in V'$, which is in the dominating solution if and only if vertex $v$ is in the vertex cover solution.
    Formally,
    \begin{align*}
        \{f(S) : S \subseteq L \ \text{s.t.} \ S \in \sol_{\textsc{VC}}\} &=
        \{\{f(v) : v \in S\} : S \subseteq V \ \text{s.t.} \ S \in \sol_{\textsc{VC}}\} \\
        &= \{\{v' : v \in S\} : S \subseteq V \ \text{s.t.} \ S \in \sol_{\textsc{VC}}\} \\
        &= \{\{v' : v' \in S' \cap f(V)\} : S' \cap f(V) \subseteq W \ \text{s.t.} \ S' \in \sol_{\textsc{DS}}\} \\
        &= \{S' \cap f(V) : S' \in \sol_{\textsc{DS}}\}.
    \end{align*}
    The correctness proof of the reduction itself is presented in \cite{DBLP:journals/corr/abs-2311-10540}.
\end{claimproof}

Finally, we have to show that the reduction is blow-up preserving.
For this, we take a look at all vertices of $V'$ and associate them with the sets $f(\U)$, $U_{\textit{on}}$ and $U_{\textit{off}}$. 

\begin{claim}
    The reduction from above from \textsc{Vertex Cover} to \textsc{Dominating Set} is an SSP reduction.
\end{claim}
\begin{claimproof}
    All vertices are mapped to their correspondence in $V'$.
    Moreover, all vertices of $\{vw_i : \{v, w\} \in E, 1 \leq i \leq |V|+1\}$ are never part of a solution of size $k'$.
    Overall, it holds that $V' = W \cup \{vw_i : \{v, w\} \in E, 1 \leq i \leq |V|+1\}$.
    Thus, we have
    \begin{align*}
        f(\U) &= f(V) = W, \\
        U_{\textit{on}} &= \emptyset, \\
        U_{\textit{off}} &= \{vw_i : \{v,w\} \in E, 1 \leq i \leq |V| + 1\}, \\
        V' &= f(V) \ \dot\cup \ U_{\textit{on}} \ \dot\cup \ U_{\textit{off}}.
    \end{align*}
\end{claimproof}

All in all, the reduction from above is a blow-up preserving SSP reduction and thus recoverable robust dominating set is $\Sigma^p_3$-complete.

\section{Conclusion}
We have shown that for a large class of NP-complete problems their corresponding recoverable robust problem is $\Sigma^p_3$-complete for several different distance measures.
For this, we introduced two new types of reductions, the blow-up SSP reduction and the blow-up preserving SSP reductions.
Blow-up SSP reductions are the basis to show $\Sigma^p_3$-completeness of a recoverable robust problem by a reduction from \textsc{3Sat} and further problems can be appended to the reduction chain by using the transitive blow-up preserving SSP reductions.
In particular, we are able to show that 24 recoverable robust problems are $\Sigma^p_3$-completeness with the ability to apply the framework to further problems.

These results and the framework might be extendible in several different dimensions.
First the natural question arises, whether these results can also be applied to most problems that are just in NP but not NP-complete, because we are only able to use this result for NP-complete problems.
The results by Jackiewicz, Kasperski and Zielińsk \cite{DBLP:journals/corr/abs-2403-20000} indicate that this might be indeed the case.
Secondly, we have only shown that these results hold for subset search problems.
Partition problems or permutation problems are not directly covered by this result.
Thus, it is of interest to the complexity of these problem classes as well.
Furthermore, it is of interest to apply the framework to more subset search problems, i.e. to show that more problems are blow-up (preserving) SSP reducible to \textsc{3Sat}.

\newpage

\bibliography{bib_general,bib_recoverable_robust,bib_reductions}

\newpage

\appendix
\section{Blow-Up SSP Reductions for Various Problems}\label{appendix:blow-up-ssp-reduction:reductions}
\begin{mdframed}[nobreak=true]
	\begin{description}
        \item[]\textsc{Satisfiability}\hfill\\
        \textbf{Instances:} Literal Set $L = \fromto{\ell_1}{\ell_n} \cup \fromto{\overline \ell_1}{\overline \ell_n}$, Clauses $C \subseteq 2^L$.\\
        \textbf{Universe:} $L =: \U$.\\
        \textbf{Solution set:} The set of all sets $L' \subseteq \U$ such that for all $i \in \fromto{1}{n}$ we have $|L' \cap \set{\ell_i, \overline \ell_i}| = 1$, and such that $|L' \cap c| \geq 1$ for all $c \in C$.
	\end{description}
\end{mdframed}

We begin with Karp's reduction from \textsc{Satisfiability} to \textsc{3-Satisfiability} \cite{DBLP:conf/coco/Karp72}.
Note that this reduction shows that RR-\textsc{Satisfiability} is also $\Sigma^p_3$-hard.

\begin{mdframed}[nobreak=true]
	\begin{description}
        \item[]\textsc{3-Satisfiability}\hfill\\
        \textbf{Instances:} Literal Set $L = \fromto{\ell_1}{\ell_n} \cup \fromto{\overline \ell_1}{\overline \ell_n}$, Clauses $C \subseteq L^{3}$.\\
        \textbf{Universe:} $L =: \U$.\\
        \textbf{Solution set:} The set of all sets $L' \subseteq \U$ such that for all $i \in \fromto{1}{n}$ we have $|L' \cap \set{\ell_i, \overline \ell_i}| = 1$, and such that $|L' \cap c| \geq 1$ for all $c \in C$.
	\end{description}
\end{mdframed}
The reduction of Karp \cite{DBLP:conf/coco/Karp72} from \textsc{Satisfiability} to \textsc{3-Satisfiability} is an SSP reduction.
Each clause of more than three literals is mapped to a set of clauses of length three by introducing helper variables $h_1, h_2, \ldots$ splitting the clause into smaller clauses.
A clause $\{a, b, c, d, \ldots\}$ with more than three literals is recursively split until there are no more clauses with more than three literals as follows:
$$
    \{a, b, c, d, \ldots\} \mapsto \{a, b, h_i\}, \{\overline h_i, c, d, \ldots\}.
$$
Thus each literal $\ell \in L$ remains in the \textsc{3Sat} instance and are one-to-one correspondent in both instances, i.e. $f(\ell) = \ell$.
Thus this reduction is indeed an SSP reduction:
\[
    \{f(S) : S \subseteq L \ \text{s.t.} \ S \in \sol_\textsc{Sat} \} = \{S' \cap f(L) : S' \subseteq L \ \text{s.t.} \ S' \in \sol_\textsc{3Sat} \}
\]

We are able to extend the reduction with the following blow-up gadget:
For a blow-up factor of $\beta_I$, we add $\beta_I$ new variables $\ell^1, \ldots, \ell^{\beta_I}$ for each existing variable $\ell \in B_{up} \subseteq L$.
By introducing additional clauses $(\ell \lor \overline \ell^i), (\overline \ell \lor \ell^i)$ for all $i \in \fromto{1}{\beta_I}$, we ensure that $\ell$ and all its corresponding blow-up literals $\ell^1, \ldots, \ell^{\beta_I}$ are logically equivalent.
Consequently if the literal $\ell$ is exchanged with $\overline \ell$, then all literals $\ell, \ell^1, \ldots, \ell^{\beta_I}$ have to be exchanged with $\overline \ell, \overline \ell^1, \ldots, \overline \ell^{\beta_I}$.

In order to calculate the $\beta_I$ for the three distance measures, we observe that a \textsc{Sat} solution always includes exactly one of the literals $\ell$ or $\overline{\ell}$.
Thus, Alice may choose all literals of $L \setminus L_b$ wrongly.
Consequently, we get the following $\beta_I$:

\begin{itemize}
    \item $\kappa$-addition: $\beta_I = |L \setminus L_b|/2$
    \item $\kappa$-deletion: $\beta_I = |L \setminus L_b|/2$
    \item Hamming distance: $\beta_I = |L \setminus L_b|$.
\end{itemize}

\begin{mdframed}[nobreak=true]
	\begin{description} 
        \item[]\textsc{Independent Set}\hfill\\
        \textbf{Instances:} Graph $G = (V,E)$, number $k \in \N$.\\
        \textbf{Universe:} Vertex set $V =: \U$.\\
        \textbf{Solution set:} The set of all independent sets of size at least $k$.
	\end{description}
\end{mdframed}
For a reduction from \textsc{3Sat} to \textsc{Independent Set}, we adapt the reduction from \textsc{3Sat} to \textsc{Vertex Cover} by Garey and Johnson \cite{DBLP:books/fm/GareyJ79}, which we already used in \Cref{sec:blow-up-ssp-reduction:reductions}.
Let $(L,C)$ be the \textsc{3Sat} instance and $G = (V,E)$ be the corresponding \textsc{Independent Set} instance.
Each literal $\ell \in L$ is mapped to a corresponding vertex $v_\ell \in V$.
Then, the vertex $v_\ell$ of $\ell$ and the vertex $v_{\overline \ell}$ of its negation $v_{\overline \ell}$ are connected by the edge $\{v_\ell,v_{\overline \ell}\} \in E$.
Next, all clauses $c \in C$ are mapped to a 3-clique.
To be more precise, for every literal $\ell \in c$, there is a vertex $v^c_\ell$, which is connected to all the other vertices $v^c_{\ell'}$ with $\ell \neq \ell' \in c$.
In the end, the vertex $v_{\overline \ell}$ is connected to $v^c_\ell$ by an edge for all clauses $c \in C$.
The threshold of the independent set is set to $k = |L|/2 + |C|$.
A picture of an example reduction can be found in \Cref{fig:reduction:3sat-independent-set}.

\tikzstyle{vertex}=[draw,circle,fill=black, minimum size=4pt,inner sep=0pt]
\tikzstyle{edge} = [draw,-]
\begin{figure}[thpb]
\centering
\resizebox{0.67\textwidth}{!}{
\begin{tikzpicture}[scale=1,auto]

\node[vertex] (x1) at (0,0) {}; \node[above] at (x1) {$v_{\ell_1}$};
\node[vertex] (notx1) at (2,0) {}; \node[above] at (notx1) {$v_{\overline \ell_1}$};
\draw[edge] (x1) to (notx1);

\node[vertex] (x2) at (4,0) {}; \node[above] at (x2) {$v_{\ell_2}$};
\node[vertex] (notx2) at (6,0) {}; \node[above] at (notx2) {$v_{\overline \ell_2}$};
\draw[edge] (x2) to (notx2);

\node[vertex] (x3) at (8,0) {}; \node[above] at (x3) {$v_{\ell_3}$};
\node[vertex] (notx3) at (10,0) {}; \node[above] at (notx3) {$v_{\overline \ell_3}$};
\draw[edge] (x3) to (notx3);

\node[vertex] (c1) at (4,-2.25) {}; \node[below] at (c1) {$v^{c_1}_{\overline \ell_1}$};
\node[vertex] (c2) at (5,-1.25) {}; \node[above right] at (c2) {$v^{c_1}_{\overline \ell_2}$};
\node[vertex] (c3) at (6,-2.25) {}; \node[below] at (c3) {$v^{c_1}_{\ell_3}$};
\node[] at (5,-1.92) {$c_1$};
\draw[edge] (c1) to (c2) to (c3) to (c1);
\draw[edge] (x1) to (c1);
\draw[edge] (x2) to (c2);
\draw[edge] (notx3) to (c3);

\node at ($(x1)+(-1,0)$) {$W$};
\draw[dashed,rounded corners] ($(x1)+(-.5,+.7)$) rectangle ($(notx3) + (.5,-.4)$);

\end{tikzpicture}
}
\caption{Classical reduction of \textsc{3Sat} to \textsc{Independent Set} for $\varphi = (\overline \ell_1 \lor \overline \ell_2 \lor \ell_3)$.}
\label{fig:reduction:3sat-independent-set}
\end{figure}
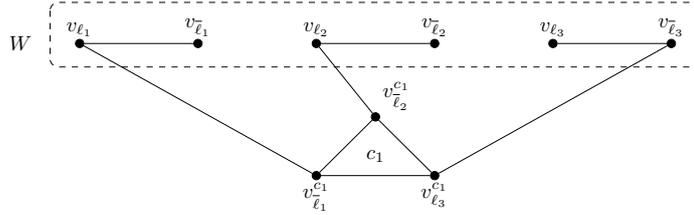

To show that this reduction is an SSP reduction, we observe that each literal $\ell$ is mapped to exactly one vertex $v_\ell$, i.e. $f(\ell) = v_\ell$.
Furthermore, it holds that literal $\ell$ is in the \textsc{3Sat} solution if and only if vertex $v_\ell$ is in the \textsc{Independent Set} solution. Thus,
\begin{align*}
    \{f(S) : S \subseteq L \ \text{s.t.} \ S \in \sol_{\textsc{3Sat}}\} &= \{f(S) : f(S) \subseteq W \ \text{s.t.} \ S \in \sol_{\textsc{3Sat}}\}\\
    &= \{S' \cap f(L) : S' \subseteq W \ \text{s.t.} \ S' \in \sol_{IS}\}\\
    &= \{S' \cap f(L) : S' \in \sol_{IS}\}.
\end{align*}
Thus, the SSP reduction is indeed correct.

Moreover, we show that a blow-up gadget exists.
This is also a complete bipartite graph $K_{\beta_I+1, \beta_I+1}$ as in the \textsc{Vertex Cover} reduction and is depicted in \Cref{fig:blow-up:3sat-independent-set}.

\tikzstyle{vertex}=[draw,circle,fill=black, minimum size=4pt,inner sep=0pt]
\tikzstyle{edge} = [draw,-]
\begin{figure}[thpb]
\centering
\begin{tikzpicture}[scale=1,auto]

\node[vertex] (x11) at (0,0) {}; \node[above] at (x11) {$v_{\ell}$};
\node[vertex] (notx11) at ($(x11) + (2,0)$) {}; \node[above] at (notx11) {$v_{\overline \ell}$};
\node[vertex] (x12) at ($(x11) + (0,1)$) {}; \node[above] at (x12) {$v_{\ell^1}$};
\node[vertex] (notx12) at ($(x11) + (2,1)$) {}; \node[above] at (notx12) {$v_{\overline \ell^1}$};
\node[vertex] (x13) at ($(x11) + (0,2)$) {}; \node[above] at (x13) {$v_{\ell^2}$};
\node[vertex] (notx13) at ($(x11) + (2,2)$) {}; \node[above] at (notx13) {$v_{\overline \ell^2}$};
\draw[edge] (x11) to (notx11);
\draw[edge] (x11) to (notx12);
\draw[edge] (x11) to (notx13);
\draw[edge] (x12) to (notx11);
\draw[edge] (x12) to (notx12);
\draw[edge] (x12) to (notx13);
\draw[edge] (x13) to (notx11);
\draw[edge] (x13) to (notx12);
\draw[edge] (x13) to (notx13);
\node[] (x11c1) at ($(x11) + (-0.5,-0.7)$) {};
\node[] (x11c2) at ($(x11) + (0,-0.7)$) {};
\node[] (x11c3) at ($(x11) + (0.5,-0.7)$) {};
\node[] (notx11c1) at ($(notx11) + (-0.5,-0.7)$) {};
\node[] (notx11c2) at ($(notx11) + (0,-0.7)$) {};
\node[] (notx11c3) at ($(notx11) + (0.5,-0.7)$) {};

\draw[-, densely dotted] (x11) to (x11c1);
\draw[-, densely dotted] (x11) to (x11c2);
\draw[-, densely dotted] (x11) to (x11c3);
\draw[-, densely dotted] (notx11) to (notx11c1);
\draw[-, densely dotted] (notx11) to (notx11c2);
\draw[-, densely dotted] (notx11) to (notx11c3);

\node at ($(x13)+(-1,0)$) {$x$};
\draw[dashed,rounded corners] ($(x11)+(-.5,-.3)$) rectangle ($(notx13) + (.5,+.6)$);

\end{tikzpicture}
\caption{Blow-up gadget for the reduction of \textsc{3Sat} to \textsc{Independent Set} with blow-up factor of $\beta = 2$.}
\label{fig:blow-up:3sat-independent-set}
\end{figure}
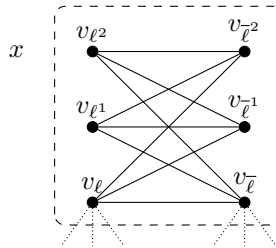

The argument is analogous to the one from the \textsc{Vertex Cover} reduction in \Cref{sec:blow-up-ssp-reduction:reductions}.
That is, either all literal vertices $v_{\ell^1}, \ldots, v_{\ell^\beta}$ or all literal vertices $v_{\overline \ell^1}, \ldots, v_{\overline \ell^\beta}$ are in the solution since the edge between $v_{\ell^i}$ and $v_{\overline \ell^j}$ for some $i \neq j$ makes the independent set invalid.
In order to calculate the $\beta_I$, we again analyse the solution structure.
Exactly one of the vertices $v_\ell$ and $v_{\overline \ell}$ for $\ell, \overline{\ell} \in L \setminus L_b$ need to be included as well as exactly one of $v^{c_j}_{\ell_1}, v^{c_j}_{\ell_2}, v^{c_j}_{\ell_3}$.
If we assume all of them are wrongly chosen, we receive the following $\beta_I$ for the three distance measures:
\begin{itemize}
    \item $\kappa$-addition: $\beta_I = |C| + |L \setminus L_b|/2$
    \item $\kappa$-deletion: $\beta_I = |C| + |L \setminus L_b|/2$
    \item Hamming distance: $\beta_I = 2|C| + |L \setminus L_b|$.
\end{itemize}

\begin{mdframed}[nobreak=true]
	\begin{description}   
        \item[]\textsc{Subset Sum}\hfill\\
        \textbf{Instances:} Numbers $\fromto{a_1}{a_n} \subseteq \N$, and target value $M \in \N$.\\
        \textbf{Universe:} $\fromto{a_1}{a_n} =: \U$.\\
        \textbf{Solution set:} The set of all sets $S \subseteq \U$ with $\sum_{a_i \in S}a_i = M$.
	\end{description}
\end{mdframed}
The reduction by Sipser \cite{DBLP:books/daglib/0086373} from \textsc{3Sat} to \textsc{Subset Sum} is an SSP reduction and be extended to a blow-up SSP reduction.
Let $(L, C)$ be the \textsc{3Sat} instance and $(\fromto{a_1}{a_n}, M)$ be the \textsc{Subset Sum} instance.
We create a table as in \Cref{fig:reduction:3sat-subset-sum} with a column for all variables and for all clauses.
We guarantee that the summation of numbers does not create carry-overs to neighboring columns by choosing a high enough (but constant) padding. 
The idea is that the target sum guarantees that exactly one literal of $\ell_i$ and $\overline \ell_i$ for the corresponding variable $x_i$ is taken into the solution.
Furthermore, each clause $c_j$ is satisfied by at least one literal such that target sum of that column of $c_j$ is reached by adding additional elements that have a one or two in that column.
Thus every number has a length of $|L|/2 + |C|$ bits.

\tikzstyle{vertex}=[draw,circle,fill=black, minimum size=4pt,inner sep=0pt]
\tikzstyle{edge} = [draw,-]
\begin{figure}[thpb]
\centering
\begin{tabular}{c|c|c|c|c|}
 & $x_1$ & $x_2$ & $x_3$ & $c_1 = \overline \ell_1 \lor \overline \ell_2 \lor \ell_3$ \\
\hline
$s_1$ & 1 & 0 & 0 & 0 \\
\hline
$s_2$ & 1 & 0 & 0 & 1 \\
\hline
$s_3$ & 0 & 1 & 0 & 0 \\
\hline
$s_4$ & 0 & 1 & 0 & 1 \\
\hline
$s_5$ & 0 & 0 & 1 & 1 \\
\hline
$s_6$ & 0 & 0 & 1 & 0 \\
\hline
$s_7$ & 0 & 0 & 0 & 1\\
\hline
$s_8$ & 0 & 0 & 0 & 2\\
\hline
$\Sigma$ & 1 & 1 & 1 & 4\\
\end{tabular}
\caption{Classical reduction of \textsc{3Sat} to \textsc{Subset Sum} for $\varphi = (\overline \ell_1 \wedge \overline \ell_2 \wedge \ell_3)$.}
\label{fig:reduction:3sat-subset-sum}
\end{figure}

Each literal $\ell_i$ is mapped to a number $a_{\ell_i}$ that contains a 1 in the column of the corresponding variable $x_i$.
Furthermore, it contains a 1 in the column of every clause $c_j$ that $\ell_i$ is contained in, i.e. $\ell_i \in c_j$.
In the end, we add two numbers for each clause $c_j$, one which has a 1 ($a^1_{c_j}$) and one which has a 2 ($a^2_{c_j}$) in the column of clause $c_j$.
We then, we define the target sum $M$ to be 1 in every variable column and 4 in every clause column.

Each literal $\ell_i$ corresponds to the number $a_{\ell_i}$, i.e. $f(\ell_i) = a_{\ell_i}$.
Because M has a 1 in every variable column, exactly one of $\ell_i$ and $\overline \ell_i$ can be part of the solution.
Furthermore, a satisfying assignment of the \textsc{Sat} instance guarantees that for each $c_j$, the column of $c_j$ can be complemented to a sum of 4 because at least one literal adds a one to that column.
On the other hand, if the target sum $M$ is reached there has to be some $a_{\ell_i}$ that induces a 1 to column $c_j$ for all $c_j \in C$.
Therefore, this reduction is an SSP reduction.

Moreover, the reduction can be extended to a blow-up SSP reduction by defining a blow-up gadget.
For this, we expand the table as depicted in \Cref{fig:blow-up:3sat-subset-sum}.

\tikzstyle{vertex}=[draw,circle,fill=black, minimum size=4pt,inner sep=0pt]
\tikzstyle{edge} = [draw,-]
\begin{figure}[thpb]
\centering
\begin{tabular}{c|c|c|c|c|c|c|c|c|c|}
 & $x_1$ & $\ell^1_1$ & $\ell^2_1$ & $\ell^3_1$  & $\overline \ell^1_1$ & $\overline \ell^2_1$ & $\overline \ell^3_1$ & \ldots & $c_1 = \overline \ell_1 \lor \overline \ell_2 \lor \ell_3$ \\
\hline
$s_1$ & 1 & 0 & 0 & 0 & 1 & 1 & 1 & \ldots & 0 \\
\hline
$s_2$ & 1 & 1 & 1 & 1 & 0 & 0 & 0 & \ldots & 1 \\
\hline
$s^1_1$ & 0 & 1 & 0 & 0 & 0 & 0 & 0 & \ldots & 0 \\
\hline
$s^2_1$ & 0 & 0 & 1 & 0 & 0 & 0 & 0 & \ldots & 0 \\
\hline
$s^3_1$ & 0 & 0 & 0 & 1 & 0 & 0 & 0 & \ldots & 0 \\
\hline
$s^1_2$ & 0 & 0 & 0 & 0 & 1 & 0 & 0 & \ldots & 0 \\
\hline
$s^2_2$ & 0 & 0 & 0 & 0 & 0 & 1 & 0 & \ldots & 0 \\
\hline
$s^3_2$ & 0 & 0 & 0 & 0 & 0 & 0 & 1 & \ldots & 0 \\
\hline
$s_3$ & 0 & 0 & 0 & 0 & 0 & 0 & 0 & \ldots & 1 \\
\hline
$s_4$ & 0 & 0 & 0 & 0 & 0 & 0 & 0 & \ldots & 2 \\
\hline
$\Sigma$ & 1 & 1 & 1 & 1 & 1 & 1 & 1 & \ldots & 4\\
\end{tabular}
\caption{Blow-up gadget for the reduction of \textsc{3Sat} to \textsc{Subset Sum} with blow up $\beta_I = 3$.}
\label{fig:blow-up:3sat-subset-sum}
\end{figure}

Let $\beta_I$ be the blow-up factor.
For every literal $\ell$, we add $\beta_I$ additional columns representing a copy of $\ell$.
We then take the number $a_{\ell_i}$ and add 1s to the $\beta_I$ newly introduced columns.
That also means that the number $a_{\ell_i}$ now has a 0 in the $\beta_I$ columns of $\overline \ell_i$.
The target sum is extended with 1s for each of the $\beta_I|L|$ newly introduced columns.
In order to fill up the 0s for every number, we add $\beta_I \cdot |L|$ new numbers $b^\ell_i$, $1 \leq i \leq \beta_I$, having exactly one 1 in each of the newly introduced columns for every literal $\ell$.
The logic of the original reduction is still valid, that is $f(\ell) = a_{\ell_i}$.
However, if $a_{\ell}$ is taken into the solution, $\beta_I$ many numbers $b^\ell_i$, $1 \leq i \leq \beta_I$, have also taken into the solution to fill up the newly introduced 0 columns.
Thus if the variable assignment shall be changed, $\beta_I + 1$ numbers ($a_{\ell_i}$ and the $b^\ell_i$, $1 \leq i \leq \beta_I$,) have to be exchanged inducing a large distance between the solutions.

For computing the $\beta_I$ for the three distance measures, we note that the literal numbers $a_{\ell_i}$ might be wrongly chosen for all $L \setminus L_b$ as well as the numbers $a^1_{c_j}$ and $a^2_{c_j}$ for $c_j \in C$, whereby one of $a^1_{c_j}$ and $a^2_{c_j}$ is always in the solution.
Thus, we have:
\begin{itemize}
    \item $\kappa$-addition: $\beta_I = |C| + |L \setminus L_b|/2$
    \item $\kappa$-deletion: $\beta_I = |C| + |L \setminus L_b|/2$
    \item Hamming distance: $\beta_I = 2|C| + |L \setminus L_b|$.
\end{itemize}

\begin{mdframed}[nobreak=true]
	\begin{description}
        \item[]\textsc{Directed Hamiltonian Path}\hfill\\
        \textbf{Instances:} Directed Graph $G = (V, A)$, Vertices $s, t \in V$.\\
        \textbf{Universe:} Arc set $A =: \U$.\\
        \textbf{Solution set:} The set of all sets $C \subseteq A$ forming a Hamiltonian path going from $s$ to $t$.
	\end{description}
\end{mdframed}

The reduction from \textsc{3Sat} to \textsc{Directed Hamiltonian Path} from Arora and Barak \cite{DBLP:books/daglib/0023084} is an SSP reduction and can be extended to a blow-up SSP reduction.
We transform the \textsc{3Sat} instance $(L, C)$ to the \textsc{Directed Hamiltonian Path} instance $G = (V, A)$ as depicted in \Cref{fig:reduction:3sat-directed-hamiltonian-path}.
We transform each pair of literals $(\ell_i, \overline \ell_i)$ (or variable $x_i$) to a path of vertices $v^1_i, \ldots, v^{4|C|}_i$.
The arcs between vertices $v^1_i, \ldots, v^{4|C|}_i$ go in either direction.
If the a solution goes from $v^1_i$ to $v^{4|C|}_i$ it represents taking literal $\ell_i$ into the solution, and if a solution goes from $v^{4|C|}_i$ to $v^1_i$, it represents taking literal $\overline \ell_i$ into the solution.
For each clause $c_j \in C$ a vertex $v_{c_j}$ is introduced and connect it to the paths induced by the literals as follows.
For a non-negated literal $\ell_i$ in clause $c_j$, we add arcs $(v^{4j-1}_i, c_j)$ and $(c_j, v^{4j-2}_i)$.
If the literal is negated, however, we add for $\overline \ell_i \in c_j$ the arcs $(v^{4j-2}, c_j)$ and $(c_j, v^{4j-1})$.
At last, we introduce the start vertex $s$ and target vertex $t$, which we connect to the graph with arcs $(s, v^1_1)$ and $(s, v^{4|C|}_1)$, $(v^1_{|L|/2}, t)$ and $(v^{4|C|}_{|L|/2}, t)$ as well as $(v^1_i, v^{4|C|}_{i+1})$ and $(v^{4|C|}_{i}, v^1_{i+1})$, for all $i \in \fromto{1}{{|L|/2}-1}$.

A Hamiltonian path from $s$ to $t$ has to go over all literal paths in exactly one direction, while the clause vertices $v_{c_j}$ are taken, whenever it is possible.
Thus, the arc $(v^1_i, v^2_i)$ represents literal $\ell_i$ and the reversed arc $(v^2_i, v^1_i)$ represents the literal $\overline \ell_i$, i.e. $f(\ell_i) = (v^1_i, v^2_i)$ and $f(\overline \ell_i) = (v^2_i, v^1_i)$.

\tikzstyle{vertex}=[draw,circle,fill=black, minimum size=4pt,inner sep=0pt]
\tikzstyle{edge} = [draw,-]
\tikzstyle{arc} = [draw,->]
\tikzstyle{doublearc} = [draw,<->]
\begin{figure}[thpb]
\centering
\begin{tikzpicture}[scale=1,auto]

\node[vertex] (s) at (0,5) {}; \node[above] at (s) {$s$};
\node[vertex] (t) at (0,1) {}; \node[below] at (t) {$t$};

\node[] (x1s) at (-2.5,4) {}; \node[right] at (x1s) {$x_1$};
\node[vertex] (x11) at ($(x1s) + (1,0)$) {};
\node[vertex] (x12) at ($(x1s) + (2,0)$) {};
\node[vertex] (x13) at ($(x1s) + (3,0)$) {};
\node[vertex] (x14) at ($(x1s) + (4,0)$) {};
\draw[doublearc] (x11) to (x12);
\draw[doublearc] (x12) to (x13);
\draw[doublearc] (x13) to (x14);

\node[] (x2s) at (-2.5,3) {}; \node[right] at (x2s) {$x_2$};
\node[vertex] (x21) at ($(x2s) + (1,0)$) {};
\node[vertex] (x22) at ($(x2s) + (2,0)$) {};
\node[vertex] (x23) at ($(x2s) + (3,0)$) {};
\node[vertex] (x24) at ($(x2s) + (4,0)$) {};
\draw[doublearc] (x21) to (x22);
\draw[doublearc] (x22) to (x23);
\draw[doublearc] (x23) to (x24);

\node[] (x3s) at (-2.5,2) {}; \node[right] at (x3s) {$x_3$};
\node[vertex] (x31) at ($(x3s) + (1,0)$) {};
\node[vertex] (x32) at ($(x3s) + (2,0)$) {};
\node[vertex] (x33) at ($(x3s) + (3,0)$) {};
\node[vertex] (x34) at ($(x3s) + (4,0)$) {};
\draw[doublearc] (x31) to (x32);
\draw[doublearc] (x32) to (x33);
\draw[doublearc] (x33) to (x34);

\draw[arc] (s) to (x11);
\draw[arc] (s) to (x14);
\draw[arc] (x11) to (x21);
\draw[arc] (x11) to (x24);
\draw[arc] (x14) to (x21);
\draw[arc] (x14) to (x24);
\draw[arc] (x21) to (x31);
\draw[arc] (x21) to (x34);
\draw[arc] (x24) to (x31);
\draw[arc] (x24) to (x34);
\draw[arc] (x31) to (t);
\draw[arc] (x34) to (t);

\node[vertex] (c1) at (4,3) {}; \node[right] at (c1) {$c_1 = \overline \ell_1 \lor \overline \ell_2 \lor \ell_3$};
\draw[arc, bend left] (x13) to (c1);
\draw[arc, bend right] (c1) to (x12);
\draw[arc, bend left] (x23) to (c1);
\draw[arc, bend right] (c1) to (x22);
\draw[arc, bend right] (x32) to (c1);
\draw[arc, bend left] (c1) to (x33);


\end{tikzpicture}
\caption{Classical reduction of \textsc{3Sat} to \textsc{Directed Hamiltonian Path} for $\varphi = (\overline \ell_1 \lor \overline \ell_2 \lor \ell_3)$.}
\label{fig:reduction:3sat-directed-hamiltonian-path}
\end{figure}

We can also find a blow-up gadget for this reduction.
The idea is to lengthen the path by $\beta_I$ additional vertices.
Thus, we also receive $\beta_I$ additional arcs.
In \Cref{fig:blow-up:3sat-directed-hamiltonian-path}, an example can be found, where $b_1^{1}, b_1^{2}$ and $b_1^{3}$ are the blow-up vertices, which also introduce the additional arcs.

\tikzstyle{vertex}=[draw,circle,fill=black, minimum size=4pt,inner sep=0pt]
\tikzstyle{edge} = [draw,-]
\tikzstyle{arc} = [draw,->, bend left]
\tikzstyle{backarc} = [draw,<-, bend right]
\tikzstyle{doublearc} = [draw,<->]
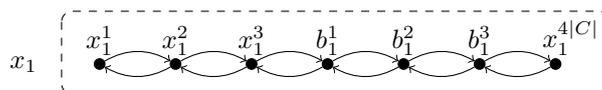
\begin{figure}[thpb]
\centering
\begin{tikzpicture}[scale=1,auto]

\node[vertex] (x11) at (-2.5,4) {}; \node[above] at (x11) {$x^{1}_1$};
\node[vertex] (x12) at ($(x11) + (1,0)$) {}; \node[above] at (x12) {$x^{2}_1$};
\node[vertex] (x13) at ($(x11) + (2,0)$) {}; \node[above] at (x13) {$x^{3}_1$};
\node[vertex] (x131) at ($(x11) + (3,0)$) {}; \node[above] at (x131) {$b^{1}_1$};
\node[vertex] (x132) at ($(x11) + (4,0)$) {}; \node[above] at (x132) {$b^{2}_1$};
\node[vertex] (x133) at ($(x11) + (5,0)$) {}; \node[above] at (x133) {$b^{3}_1$};
\node[vertex] (x14) at ($(x11) + (6,0)$) {}; \node[above] at ($(x14)+(0.2,0)$) {$x^{4|C|}_1$};
\draw[arc] (x11) to (x12);
\draw[arc] (x12) to (x13);
\draw[arc] (x13) to (x131);
\draw[arc] (x131) to (x132);
\draw[arc] (x132) to (x133);
\draw[arc] (x133) to (x14);
\draw[backarc] (x11) to (x12);
\draw[backarc] (x12) to (x13);
\draw[backarc] (x13) to (x131);
\draw[backarc] (x131) to (x132);
\draw[backarc] (x132) to (x133);
\draw[backarc] (x133) to (x14);

\node at ($(x11)+(-1,0)$) {$x_1$};
\draw[dashed,rounded corners] ($(x11)+(-.5,+.7)$) rectangle ($(x14) + (.7,-.4)$);

\end{tikzpicture}
\caption{Blow-up gadget for the reduction of \textsc{3Sat} to \textsc{Directed Hamiltonian Path} with blow up $\beta_I = 3$.}
\label{fig:blow-up:3sat-directed-hamiltonian-path}
\end{figure}

For computing the $\beta_I$ for the three distance measures, the solution can be wrong on all literals in $L \setminus L_b$ and in this case the clause vertices $c_j$ need to be included in the Hamiltonian path via a different literal path.
Furthermore, the arcs connecting the literals as well as vertices $s$ and $t$ need to be changed for the literals from $L \setminus L_b$.
Thus, we get the following $\beta_I$:
\begin{itemize}
    \item $\kappa$-addition: $\beta_I = 2|C| + (4|C|+2) \cdot |L \setminus L_b|/2$
    \item $\kappa$-deletion: $\beta_I = 2|C| + (4|C|+2) \cdot |L \setminus L_b|/2$
    \item Hamming distance: $\beta_I = 4|C| + (8|C|+4) \cdot |L \setminus L_b|/2$.
\end{itemize}

\begin{mdframed}[nobreak=true]
	\begin{description}
        \item[]\textsc{Directed Two Disjoint Path}\hfill\\
        \textbf{Instances:} Graph $G = (V, A)$, $s_i, t_i \in V$ for $i \in \{1, 2\}$.\\
        \textbf{Universe:} Arc set $A =: \U$.\\
        \textbf{Solution set:} The set of all disjoint paths $(P_1, P_2)$ such that $P_i$ is the path from $s_i$ to $t_i$ for $i \in \{1, 2\}$.
	\end{description}
\end{mdframed}

The reduction by Fortune, Hopcroft and Wyllie \cite{DBLP:journals/tcs/FortuneHW80} is an SSP reduction, which is also extendable to a blow-up SSP reduction.
The reduction makes extensive use of a switch gadget, which is depicted in \Cref{fig:reduction:3sat-directed-two-disjoint-path:switch}.
The gadget has four input arcs, $B,C,W$ and $Y$, and four output arcs, $A,D,X$ and $Z$.
The idea of this switch gadget is that if you have two disjoint paths going through the gadget entering at $B$ and $C$, then the path entering at $B$ must leave at $D$ and the one entering at $C$ must leave at $A$, and additionally either a path from $W$ to $X$ exists or a path from $Y$ to $Z$ exists.
We can then use the first of the two paths to run first through the switches and then to the rest of the construction and the second path to run through the switches as in visualized in \Cref{fig:reduction:3sat-directed-two-disjoint-path}.
In doing so, the second path running only through the switches controls that the first path running through the construction is only able to satisfy the clauses according to the literals in the solution.

\tikzstyle{vertex}=[draw,circle,fill=black, minimum size=4pt,inner sep=0pt]
\tikzstyle{edge} = [draw,-]
\tikzstyle{arc} = [draw,->]
\tikzstyle{doublearc} = [draw,<->]
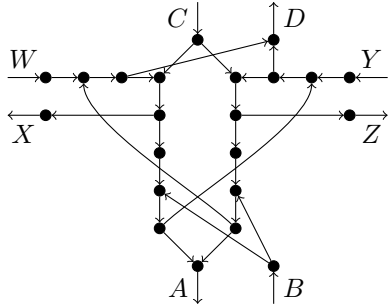
\begin{figure}[thpb]
\centering
\begin{tikzpicture}[scale=0.5,auto]

\node[vertex] (c) at (0,0) {};
\node[vertex] (a) at (0,-6) {};
\draw[arc] ($(c) + (0,1)$) to node[left] {$C$} (c);
\draw[arc] (a) to node[left] {$A$} ($(a) - (0,1)$);

\node[vertex] (01) at ($(c) + (-1,-1)$) {};
\node[vertex] (02) at ($(c) + (-1,-2)$) {};
\node[vertex] (03) at ($(c) + (-1,-3)$) {};
\node[vertex] (04) at ($(c) + (-1,-4)$) {};
\node[vertex] (05) at ($(c) + (-1,-5)$) {};
\draw[arc] (c) to (01);
\draw[arc] (01) to (02);
\draw[arc] (02) to (03);
\draw[arc] (03) to (04);
\draw[arc] (04) to (05);
\draw[arc] (05) to (a);

\node[vertex] (11) at ($(c) + (1,-1)$) {};
\node[vertex] (12) at ($(c) + (1,-2)$) {};
\node[vertex] (13) at ($(c) + (1,-3)$) {};
\node[vertex] (14) at ($(c) + (1,-4)$) {};
\node[vertex] (15) at ($(c) + (1,-5)$) {};
\draw[arc] (c) to (11);
\draw[arc] (11) to (12);
\draw[arc] (12) to (13);
\draw[arc] (13) to (14);
\draw[arc] (14) to (15);
\draw[arc] (15) to (a);

\node[vertex] (08) at (-4,-1) {};
\node[vertex] (09) at (-3,-1) {};
\node[vertex] (00) at (-2,-1) {};
\draw[arc] ($(08) - (1,0)$) to node[above] {$W$} (08);
\draw[arc] (08) to (09);
\draw[arc] (09) to (00);
\draw[arc] (00) to (01);
\draw[arc, out=135, in=270, looseness=0.5] (15) to (09);

\node[vertex] (18) at (4,-1) {};
\node[vertex] (19) at (3,-1) {};
\node[vertex] (10) at (2,-1) {};
\draw[arc] ($(18) + (1,0)$) to node[above] {$Y$} (18);
\draw[arc] (18) to (19);
\draw[arc] (19) to (10);
\draw[arc] (10) to (11);
\draw[arc, out=45, in=270, looseness=0.5] (05) to (19);

\node[vertex] (b) at (2,-6) {};
\draw[arc] (b) to (04);
\draw[arc] (b) to (14);
\draw[arc] ($(b) - (0,1)$) to node[right] {$B$} (b);

\node[vertex] (d) at (2,0) {};
\draw[arc] (00) to (d);
\draw[arc] (10) to (d);
\draw[arc] (d) to node[right] {$D$} ($(d) + (0,1)$);

\node[vertex] (011) at (-4,-2) {};
\draw[arc] (02) to (011);
\draw[arc] (011) to node[below] {$X$} ($(011) - (1,0)$);

\node[vertex] (111) at (4,-2) {};
\draw[arc] (12) to (111);
\draw[arc] (111) to node[below] {$Z$} ($(111) + (1,0)$);

\end{tikzpicture}
\caption{The switch gadget.}
\label{fig:reduction:3sat-directed-two-disjoint-path:switch}
\end{figure}

Let $(L, C)$ be the \textsc{3Sat} instance and $G = (V, A, s_1, t_1, s_2, t_2)$ be the resulting \textsc{Directed Two Disjoint Path} instance.
First, we introduce four vertices $s_1, t_1, s_2, t_2$ representing the start and ends of the two disjoint paths.
For every literal $\ell$, we create a path $\ell^1, \ldots, \ell^{4|C|}$ of $4|C|$ vertices.
For every literal pair $\ell_i, \overline \ell_i$ (or variable $x_i$), we connect the paths by introducing two additional vertices $x^s_i$ and $x^t_i$ with arcs $(x^s_i, \ell^1_i), (x^s_i, \overline \ell^1_i)$ and $(\ell^{4|C|}_i, x^t_i), (\overline \ell^{4|C|}_i, x^t_i)$.
We connect the literal gadgets by adding the arcs $(x^t_i, x^t_{i+1})$ for all $i \in \fromto{1}{|L|/2-1}$.
For each clause $c_j \in C$, we add two vertices $c^1_j$ and $c^2_j$ and connect them by three arcs.
We connect these clause vertices with the arcs $(c^2_{|C|}, t_1)$ as well as $(c^2_j, c^1_{j+1})$ for $j \in \fromto{1}{|C|-1}$.
At last, we connect the literal paths with the clause path with an arc $(x^t_{|L|/2}, c^1_1)$.

Now, we have everything to introduce the switches into the construction.
We stack the switches one after another by merging the $C$ and $D$ input arc and the $A$ and $B$ input arcs, respectively.
Then, we connect $s_2$ to the input arc $C$ of the last switch in the stack and $t_2$ to the output arc $A$ of the first switch of the complete switch stack.
We do this analogously for $s_1$, which we connect to the input arc $B$ of the first switch of the stack and the rest of the construction with the output arc $D$ of the last switch of the stack.
Thus both path run through the switch stack as described above.
At last, we use the switches to check whether the \textsc{Sat} assignment is correct.
For this, we use the schematic description of a switch as depicted in \Cref{fig:reduction:3sat-directed-two-disjoint-path:switch:schema}.

\tikzstyle{vertex}=[draw,circle,fill=black, minimum size=4pt,inner sep=0pt]
\tikzstyle{edge} = [draw,-]
\tikzstyle{arc} = [draw,->]
\tikzstyle{doublearc} = [draw,<->]
\begin{figure}[thpb]
\centering
\begin{tikzpicture}[scale=0.5,auto]

\node[vertex, label=above:{}] (w) at (0,0) {};
\node[vertex, label=below:{}] (x) at (0,-2) {};
\node[vertex, label=above:{}] (y) at (2,0) {};
\node[vertex, label=below:{}] (z) at (2,-2) {};

\draw[arc] (w) to (x);
\draw[arc] (y) to (z);
\draw[edge] (0,-1) to (2,-1);

\end{tikzpicture}
\caption{The schematic switch gadget.}
\label{fig:reduction:3sat-directed-two-disjoint-path:switch:schema}
\end{figure}
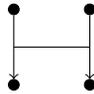

The arc $(\overline \ell^{4j-2}_i, \overline \ell^{4j-1}_i)$ is connected to $(c^1_j, c^2_j)$ by a switch if and only if the corresponding literal $\ell_i \in L$ is in clause $c_j \in C$.
More precisely, the arc $(\overline \ell^{4j-2}_i, \overline \ell^{4j-1}_i)$ is substituted by using the input arc $W$ from $\overline \ell^{4j-2}_i$ and output arc $X$ to $\overline \ell^{4j-1}_i$ and for the clause vertices $c^1_j$ is incident to input arc $Y$ and output arc $Z$ is incident to $c^2_j$.
Because only one path either from $W$ to $X$ or from $Y$ to $Z$ is usable, the path has to go from $s_1$ through the switch stack, then over the literal paths of the literals that are in the solution and at last over the clause vertices to $t_1$.
If for a clause there is no literal satisfying it, the path in the switch is blocked.
The full construction can be found in \Cref{fig:reduction:3sat-directed-two-disjoint-path}.

\tikzstyle{vertex}=[draw,circle,fill=black, minimum size=4pt,inner sep=0pt]
\tikzstyle{edge} = [draw,-]
\tikzstyle{arc} = [draw,->]
\tikzstyle{doublearc} = [draw,<->]
\begin{figure}[thpb]
\centering
\begin{tikzpicture}[scale=0.5,auto]

\node[vertex] (x1s) at (1,0) {};
\node[vertex] (x101) at ($(x1s) + (1,1)$) {};
\node[vertex] (x102) at ($(x1s) + (2,1)$) {};
\node[vertex] (x103) at ($(x1s) + (3,1)$) {};
\node[vertex] (x104) at ($(x1s) + (4,1)$) {}; \node[above left] at (x101) {$\overline \ell_1$};
\node[vertex] (x111) at ($(x1s) + (1,-1)$) {};
\node[vertex] (x112) at ($(x1s) + (2,-1)$) {};
\node[vertex] (x113) at ($(x1s) + (3,-1)$) {};
\node[vertex] (x114) at ($(x1s) + (4,-1)$) {}; \node[below left] at (x111) {$\ell_1$};
\node[vertex] (x1t) at ($(x1s) + (5,0)$) {};
\draw[arc] (x1s) to (x101);
\draw[arc] (x101) to (x102);
\draw[arc] (x102) to (x103);
\draw[arc] (x103) to (x104);
\draw[arc] (x104) to (x1t);
\draw[arc] (x1s) to (x111);
\draw[arc] (x111) to (x112);
\draw[arc] (x112) to (x113);
\draw[arc] (x113) to (x114);
\draw[arc] (x114) to (x1t);

\node[vertex] (x2s) at (7,0) {};
\node[vertex] (x201) at ($(x2s) + (1,1)$) {};
\node[vertex] (x202) at ($(x2s) + (2,1)$) {};
\node[vertex] (x203) at ($(x2s) + (3,1)$) {};
\node[vertex] (x204) at ($(x2s) + (4,1)$) {}; \node[above left] at (x201) {$\overline \ell_2$};
\node[vertex] (x211) at ($(x2s) + (1,-1)$) {};
\node[vertex] (x212) at ($(x2s) + (2,-1)$) {};
\node[vertex] (x213) at ($(x2s) + (3,-1)$) {};
\node[vertex] (x214) at ($(x2s) + (4,-1)$) {}; \node[below left] at (x211) {$\ell_2$};
\node[vertex] (x2t) at ($(x2s) + (5,0)$) {};
\draw[arc] (x2s) to (x201);
\draw[arc] (x201) to (x202);
\draw[arc] (x202) to (x203);
\draw[arc] (x203) to (x204);
\draw[arc] (x204) to (x2t);
\draw[arc] (x2s) to (x211);
\draw[arc] (x211) to (x212);
\draw[arc] (x212) to (x213);
\draw[arc] (x213) to (x214);
\draw[arc] (x214) to (x2t);

\node[vertex] (x3s) at (13,0) {};
\node[vertex] (x301) at ($(x3s) + (1,1)$) {};
\node[vertex] (x302) at ($(x3s) + (2,1)$) {};
\node[vertex] (x303) at ($(x3s) + (3,1)$) {};
\node[vertex] (x304) at ($(x3s) + (4,1)$) {}; \node[above left] at (x301) {$\overline \ell_3$};
\node[vertex] (x311) at ($(x3s) + (1,-1)$) {};
\node[vertex] (x312) at ($(x3s) + (2,-1)$) {};
\node[vertex] (x313) at ($(x3s) + (3,-1)$) {};
\node[vertex] (x314) at ($(x3s) + (4,-1)$) {}; \node[below left] at (x311) {$\ell_3$};
\node[vertex] (x3t) at ($(x3s) + (5,0)$) {};
\draw[arc] (x3s) to (x301);
\draw[arc] (x301) to (x302);
\draw[arc] (x302) to (x303);
\draw[arc] (x303) to (x304);
\draw[arc] (x304) to (x3t);
\draw[arc] (x3s) to (x311);
\draw[arc] (x311) to (x312);
\draw[arc] (x312) to (x313);
\draw[arc] (x313) to (x314);
\draw[arc] (x314) to (x3t);

\node[vertex] (s1) at (-3,0) {}; \node[above left] at (s1) {$s_1$};
\node[vertex] (b0) at (-2,0) {}; \node[above] at (b0) {$B_1$};
\node[vertex] (d1) at (0,0) {}; \node[below] at (d1) {$D_{\ell\textit{ast}}$};
\node[vertex] (t11) at (19,0) {}; \node[above right] at (t11) {};
\node[vertex] (t1) at (0,-4) {}; \node[above left] at (t1) {$t_1$};
\node[vertex] (s11) at (19,-4) {}; \node[above right] at (s11) {};
\draw[arc] (s1) to (b0);
\draw[arc] (d1) to (x1s);
\draw[arc] (x1t) to (x2s);
\draw[arc] (x2t) to (x3s);
\draw[arc] (x3t) to (t11);
\draw[arc, out=0, in=0, looseness=1] (t11) to (s11);

\node[vertex] (s2) at (6.5,4) {}; \node[above] at (s2) {$s_2$};
\node[vertex] (t2) at (12.5,4) {}; \node[above] at (t2) {$t_2$};
\node[vertex] (c1) at (8,4) {}; \node[below] at (c1) {$C_{\ell\textit{ast}}$};
\node[vertex] (a0) at (11,4) {}; \node[below] at (a0) {$A_1$};
\draw[arc] (s2) to (c1);
\draw[arc] (a0) to (t2);

\node[vertex] (c11) at (8,-4) {}; \node[below] at (c11) {$c^1_1$};
\node[below] (c1) at (9.5,-5) {$c_1$};
\node[vertex] (c12) at (11,-4) {}; \node[below] at (c12) {$c^2_1$};
\draw[arc] (s11) to (c12);
\draw[arc] (c11) to (t1);

\draw[arc] (c12) to (c11);
\draw[arc, in=45, out=135, looseness=0.7] (c12) to (c11);
\draw[arc, in=315, out=225, looseness=0.7] (c12) to (c11);

\draw[edge] (9.5,-4) to ($(x112)+(0.4,0)$);
\draw[edge] (9.5,-3.5) to ($(x212)+(0.4,0)$);
\draw[edge] (9.5,-4.5) to ($(x302)+(0.4,0)$);

\end{tikzpicture}
\caption{Classical reduction of \textsc{3Sat} to \textsc{Directed Two Disjoint Path} for $\varphi = (\overline \ell_1 \lor \overline \ell_2 \lor \ell_3)$.}
\label{fig:reduction:3sat-directed-two-disjoint-path}
\end{figure}
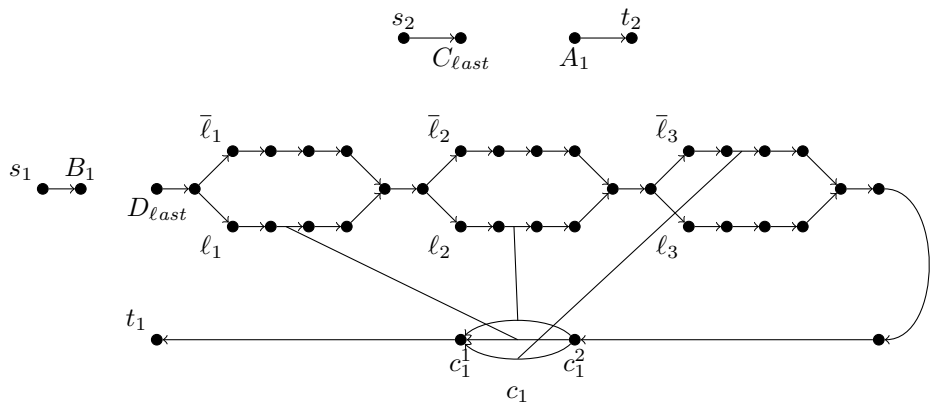

There is a one-to-one correspondence between the literals and the arcs of the path from $s_1$ to $t_1$.
We can define $f(\ell_i) = (x^s_i, \ell^1_i)$ because the path over $(x^s_i, \ell^1_i)$ is taken if and only if $\ell_i$ is in the \textsc{Sat} solution.

The blow-up gadget for this reduction works analogously as the one in the \textsc{Directed Hamiltonian Cycle}.
We introduce $\beta_I$ additional vertices (and therefore arcs) to each path corresponding to a literal $\ell \in L$.
Thus, the functionality of the gadgets is not impaired and at least $\beta_I + 1$ many arcs have to be exchanged to achieve a reassignment of literal $\ell$ to $\overline \ell$.
An example can be found in \Cref{fig:blow-up:3sat-directed-two-disjoint-path}.
For an analysis of the $\beta_I$, we need to closely consider both paths from $s_1$ to $t_1$ and $s_2$ to $t_2$ as well as the switch gadgets.
Again all literals in $L \setminus L_b$ might be chosen incorrectly such that Alice needs to be able to recover from this.
It is easy to see that the path corresponding to that literal need to be changed, these are $4|C|$ arcs.
Additionally, the switch gadgets need to be run through differently, which are up to $5|C|$ arcs for each literal. 
Furthermore, the clauses may need to be passed over different arcs, which are actually 5 arcs in the switch gadgets.
These are again up to $5|C|$ arcs.
Now, we need to consider the switch gadgets in the path from $s_2$ to $t_2$ and in the first half of the path from $s_1$ to $t_1$.
If a literal needs to be changed or a clause needs to be passed on different way (the path from $W$ to $X$ is exchanged with $Y$ to $Z$), the switch gadgets need be run through in a different way.
Thus, the path from $B$ to $D$ and the path from $A$ to $C$ needs to be mirrored and thus exchanged completely.
For every literal from $L \setminus L_b$, this might happen $|C|$ times and $12$ arcs need to be exchanged.
On the other hand for every clause $c_j \in C$, this might happen two times that $12$ arcs need to be exchanged.
Overall, we get the following $\beta_I$:
\begin{itemize}
    \item $\kappa$-addition: $\beta_I = (12+5)|C| + (12|C|+5|C|+4|C|) \cdot |L \setminus L_b|/2$
    \item $\kappa$-deletion: $\beta_I = 17|C| + 21|C| \cdot |L \setminus L_b|/2$
    \item Hamming distance: $\beta_I = 34|C| + 42|C| \cdot |L \setminus L_b|/2$.
\end{itemize}

\tikzstyle{vertex}=[draw,circle,fill=black, minimum size=4pt,inner sep=0pt]
\tikzstyle{edge} = [draw,-]
\tikzstyle{arc} = [draw,->]
\tikzstyle{backarc} = [draw,<-, bend right]
\tikzstyle{doublearc} = [draw,<->]
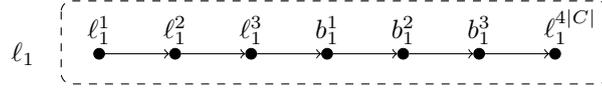
\begin{figure}[thpb]
\centering
\begin{tikzpicture}[scale=1,auto]

\node[vertex] (x11) at (-2.5,4) {}; \node[above] at (x11) {$\ell^{1}_1$};
\node[vertex] (x12) at ($(x11) + (1,0)$) {}; \node[above] at (x12) {$\ell^{2}_1$};
\node[vertex] (x13) at ($(x11) + (2,0)$) {}; \node[above] at (x13) {$\ell^{3}_1$};
\node[vertex] (x131) at ($(x11) + (3,0)$) {}; \node[above] at (x131) {$b^{1}_1$};
\node[vertex] (x132) at ($(x11) + (4,0)$) {}; \node[above] at (x132) {$b^{2}_1$};
\node[vertex] (x133) at ($(x11) + (5,0)$) {}; \node[above] at (x133) {$b^{3}_1$};
\node[vertex] (x14) at ($(x11) + (6,0)$) {}; \node[above] at ($(x14)+(0.2,0)$) {$\ell^{4|C|}_1$};
\draw[arc] (x11) to (x12);
\draw[arc] (x12) to (x13);
\draw[arc] (x13) to (x131);
\draw[arc] (x131) to (x132);
\draw[arc] (x132) to (x133);
\draw[arc] (x133) to (x14);

\node at ($(x11)+(-1,0)$) {$\ell_1$};
\draw[dashed,rounded corners] ($(x11)+(-.5,+.7)$) rectangle ($(x14) + (.7,-.4)$);

\end{tikzpicture}
\caption{Blow-up gadget for the reduction of \textsc{3Sat} to \textsc{Directed Two Disjoint Path} with blow up $\beta_I = 3$.}
\label{fig:blow-up:3sat-directed-two-disjoint-path}
\end{figure}

\begin{mdframed}[nobreak=true]
    \begin{description}
        \item[]\textsc{Steiner Tree}\hfill\\
        \textbf{Instances:} Undirected graph $G = (S \cup T, E)$, set of Steiner vertices $S$, set of terminal vertices $T$, edge weights $c: E \rightarrow \N$, number $k \in \N$.\\
        \textbf{Universe:} Edge set $E =: \U$.\\
        \textbf{Solution set:} The set of all sets $E' \subseteq E$ such that $E'$ is a tree connecting all terminal vertices from $T$ with $\sum_{e' \in E'} c(e') \leq k$.
    \end{description}
\end{mdframed}
There is a folklore reduction from \textsc{3Sat} to \textsc{Steiner Tree}, which is an SSP reduction, and which is depicted in \Cref{fig:reduction:3sat-steiner-tree}.
First, there are designated terminal vertices $s$ and $t$.
For every literal $\ell \in L$, there is a Steiner vertex $v_\ell$.
Additionally for every literal pair $(\ell_i, \overline \ell_i)$, $1 \leq i \leq |L|/2-1$, we add a Steiner vertex $v_i$. We define $v_0 := s$.
Then all of the above vertices are connected into a \enquote{diamond chain}, where we begin with $s$ connected to both $v_{\ell_1}$ and $v_{\overline \ell_1}$.
Both vertices $v_{\ell_1}$ and $v_{\overline \ell_1}$ are connected to $v_1$.
This vertex $v_1$ is then connected to vertices $v_{\ell_2}$ and $v_{\overline \ell_2}$ and so on.
At last, $v_{\ell_{|L|/2}}$ and $v_{\overline \ell_{|L|/2}}$ are connected to $t$.

Furthermore, for every clause $c_j \in C$, we add a corresponding terminal vertex $v_{c_j}$.
The vertex $v_{c_j}$ is then connected its corresponding literals vertices $v_\ell$ for $\ell \in c_j$ via a path of Steiner vertices of length $|L| + 1$.
The costs of every edge is set to $1$ and the threshold is set to $k = |L| + |C| \cdot (|L| + 1)$.

For the correctness, observe that every solution of Steiner tree includes a path from $s$ to $t$ over the literal vertices because all paths over a clause vertex are longer than $|L|$.
This path encodes the set of literals included in a corresponding \textsc{3Sat} solution, where a positive literal $\ell_i$ for $1 \leq i \leq |L|/2$ is included in the \textsc{3Sat} solution if and only if the edge $\{v_{i-1}, v_{\ell_i}\}$ is in the Steiner tree solution.
The analogous statement holds for negative literals.
We therefore define the embedding function $f$ in the above fashion, i.e. for all $\ell \in L$ we have $f(\ell) = \set{v_{i-1}, \ell}$.
Next, for every clause $c_j$, the path from literal $v_\ell$ to terminal vertex $c_j$ for one $\ell \in c_j$ is included in the solution as well.
Thus, $|C|$ paths of length $|L|+1$ are included.
If a clause is not satisfied, then a path of length of at least $|L|+2$ is needed to include a terminal vertex representing one of the clauses, which violates the threshold.
Thus, the reduction is correct.

\tikzstyle{vertex}=[draw,circle,fill=black, minimum size=4pt,inner sep=0pt]
\tikzstyle{terminal}=[draw,rectangle,fill=black, minimum size=4pt,inner sep=0pt]
\tikzstyle{edge} = [draw,-]
\begin{figure}[thpb]
\centering
\begin{tikzpicture}[scale=1,auto]

\node[terminal] (s) at (0,0) {}; \node[left] at (s) {$s$};
\node[vertex] (x12) at ($(s) + (2,0)$) {};
\node[vertex] (x23) at ($(s) + (4,0)$) {};
\node[terminal] (t) at ($(s) + (6,0)$) {}; \node[right] at (t) {$t$};

\node[vertex] (x1) at ($(s) + (1,0.5)$) {}; \node[above] at (x1) {$\ell_1$};
\node[vertex] (notx1) at ($(s) + (1,-0.5)$) {}; \node[above] at (notx1) {$\overline \ell_1$};

\node[vertex] (x2) at ($(x12) + (1,0.5)$) {}; \node[above] at (x2) {$\ell_2$};
\node[vertex] (notx2) at ($(x12) + (1,-0.5)$) {}; \node[above] at (notx2) {$\overline \ell_2$};

\node[vertex] (x3) at ($(x23) + (1,0.5)$) {}; \node[above] at (x3) {$\ell_3$};
\node[vertex] (notx3) at ($(x23) + (1,-0.5)$) {}; \node[above] at (notx3) {$\overline \ell_3$};

\draw[edge] (s) to (x1);
\draw[edge] (s) to (notx1);
\draw[edge] (x1) to (x12);
\draw[edge] (notx1) to (x12);
\draw[edge] (x12) to (x2);
\draw[edge] (x12) to (notx2);
\draw[edge] (x2) to (x23);
\draw[edge] (notx2) to (x23);
\draw[edge] (x23) to (x3);
\draw[edge] (x23) to (notx3);
\draw[edge] (x3) to (t);
\draw[edge] (notx3) to (t);

\node[terminal] (c1) at (3,-2.5) {}; \node[below] at (c1) {$c_1 = \overline \ell_1 \lor \overline \ell_2 \lor \ell_3$};

\node[vertex] (11) at ($(notx1) + (0.67,-0.5)$) {};
\node (12) at ($(11) + (0,-0.4)$) {$\vdots$};
\node[vertex] (13) at ($(11) + (0,-1)$) {};

\node[vertex] (21) at ($(notx2) + (0,-0.5)$) {};
\node (22) at ($(21) + (0,-0.4)$) {$\vdots$};
\node[vertex] (23) at ($(21) + (0,-1)$) {};

\node[vertex] (31) at ($(notx3) + (-0.67,-0.5)$) {};
\node (32) at ($(31) + (0,-0.4)$) {$\vdots$};
\node[vertex] (33) at ($(31) + (0,-1)$) {};

\draw[edge] (notx1) to (11);
\draw[edge] (11) to ($(12) + (0,0.18)$);
\draw[edge] (12) to (13);
\draw[edge] (13) to (c1);

\draw[edge] (notx2) to (21);
\draw[edge] (21) to ($(22) + (0,0.18)$);
\draw[edge] (22) to (23);
\draw[edge] (23) to (c1);

\draw[edge] (x3) to (31);
\draw[edge] (31) to ($(32) + (0,0.18)$);
\draw[edge] (32) to (33);
\draw[edge] (33) to (c1);

\node at ($(s)+(-1,0)$) {$V'$};
\draw[dashed,rounded corners] ($(s)+(-0.5,+1)$) rectangle ($(t) + (0.5,-0.67)$);

\end{tikzpicture}
\caption{Classical reduction of \textsc{3Sat} to \textsc{Steiner Tree}.}
\label{fig:reduction:3sat-steiner-tree}
\end{figure}
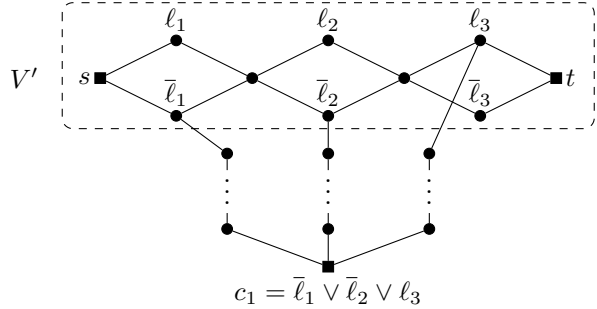

For the blow-up gadget, we add $\beta_I$ new vertices $v_{\ell^i}$, $1 \leq i \leq \beta_I$, for each literal vertex $v_\ell$, $\ell \in L_b$ to the graph.
Moreover, we connect all of these vertices $v_{\ell^i}$ to the existing $v_\ell$ with the edges $\{\{v_\ell, v_{\ell^i}\} : 1 \leq i \leq \beta_I\}$.
At last, we add an edge between each pair of $v_{\ell^i}$ and $v_{\overline \ell^i}$, for $1 \leq i \leq \beta_I$, which is subdivided by one terminal vertex $t^{\ell_i}$.
The construction is depicted in \Cref{fig:blow-up:3sat-steiner-tree}.
The threshold for the Steiner tree is increased by $2\beta_I$ for each pair $\ell, \overline \ell \in L_b$ of blow-up literals.
This blow-up gadget is correct because if vertex $v_\ell$ is connected to the Steiner tree, we are able to include the edges $\{\{v_\ell, v_{\ell^i}\}, \{v_{\ell^i}, t^{\ell_i}\} : 1 \leq i \leq \beta_I\}$, which needs $2\beta_I$ additional edges and thus lies within the threshold.
On the other hand, if an edge from $\{\{v_{\overline \ell}, v_{\overline \ell^i}\}, \{v_{\overline \ell^i}, t^{\ell_i}\} : 1 \leq i \leq \beta_I\}$ is used this edge has to be connected to the Steiner tree.
This is only doable by connecting $v_{\overline \ell}$ to the Steiner tree, which introduces additional cost of one or by connecting it over one of the terminal vertices $t^{\ell_i}$.
However in the last case, this implies that $t^{\ell_i}$ is already connected to the Steiner tree.
Thus, it is an unnecessary edge that does not connect any terminal vertex of the Steiner tree and additional costs of one are introduced.
Because each solution is an optimal Steiner tree, this is a contradiction.
Because the blow-up gadget does not interfere with the functionality of the original reduction, this blow-up gadget is correct.

\tikzstyle{vertex}=[draw,circle,fill=black, minimum size=4pt,inner sep=0pt]
\tikzstyle{steiner}=[draw,rectangle,fill=black, minimum size=4pt,inner sep=0pt]
\tikzstyle{edge} = [draw,-]
\begin{figure}[thpb]
\centering
\begin{tikzpicture}[scale=1,auto]

\node[vertex] (x11) at (0,0) {}; \node[above] at (x11) {$v_{\ell}$};
\node[vertex] (notx11) at ($(x11) + (2,0)$) {}; \node[above] at (notx11) {$v_{\overline \ell}$};
\node[vertex] (x12) at ($(x11) + (0,1)$) {}; \node[above] at (x12) {$v_{\ell^1}$};
\node[vertex] (notx12) at ($(x11) + (2,1)$) {}; \node[above] at (notx12) {$v_{\overline \ell^1}$};
\node[vertex] (x13) at ($(x11) + (0,2)$) {}; \node[above] at (x13) {$v_{\ell^2}$};
\node[vertex] (notx13) at ($(x11) + (2,2)$) {}; \node[above] at (notx13) {$v_{\overline \ell^2}$};
\draw[edge,out=180,in=180] (x11) to (x12);
\draw[edge,out=180,in=180] (x11) to (x13);
\draw[edge,out=0,in=0] (notx11) to (notx12);
\draw[edge,out=0,in=0] (notx11) to (notx13);

\draw[edge] (x12) to (notx12);
\draw[edge] (x13) to (notx13);

\node[steiner] (s1) at (1,1) {}; \node[above] at (s1) {$t^{\ell_1}$};
\node[steiner] (s2) at (1,2) {}; \node[above] at (s2) {$t^{\ell_2}$};

\node[] (x11c1) at ($(x11) + (-0.5,-0.7)$) {};
\node[] (x11c2) at ($(x11) + (0,-0.7)$) {};
\node[] (x11c3) at ($(x11) + (0.5,-0.7)$) {};
\node[] (notx11c1) at ($(notx11) + (-0.5,-0.7)$) {};
\node[] (notx11c2) at ($(notx11) + (0,-0.7)$) {};
\node[] (notx11c3) at ($(notx11) + (0.5,-0.7)$) {};

\draw[-, densely dotted] (x11) to (x11c1);
\draw[-, densely dotted] (x11) to (x11c3);
\draw[-, densely dotted] (notx11) to (notx11c1);
\draw[-, densely dotted] (notx11) to (notx11c3);

\node at ($(x13)+(-1,0)$) {$x$};
\draw[dashed,rounded corners] ($(x11)+(-.75,-.3)$) rectangle ($(notx13) + (.75,+.6)$);

\end{tikzpicture}
\caption{Blow-up gadget for the reduction of \textsc{3Sat} to \textsc{Steiner Tree} with blow up $\beta = 2$.}
\label{fig:blow-up:3sat-steiner-tree}
\end{figure}
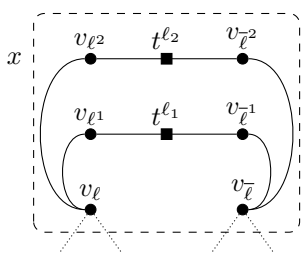

For computing the $\beta_I$ for the three distance measures, we again consider the wrongly chosen literals $L \setminus L_b$ from which Alice has to recover.
A wrongly chosen literal induces that the two edges $\{v_{i-1}, v_\ell\}$ and $\{v_\ell, v_i\}$ have to be exchanged by the edges $\{v_{i-1}, v_{\overline \ell}\}$ and $\{v_{\overline \ell}, v_i\}$.
This literal also may induce $|C| \cdot (|L|+1)$ additional edges because clauses might be connected to it.
Thus, we obtain the following $\beta_I$:
\begin{itemize}
    \item $\kappa$-addition: $\beta_I = |C| \cdot (|L|+1) + 2|L \setminus L_b|$
    \item $\kappa$-deletion: $\beta_I = |C| \cdot (|L|+1) + 2|L \setminus L_b|$
    \item Hamming distance: $\beta_I = 2|C| \cdot (|L|+1) + 4|L \setminus L_b|$.
\end{itemize}

\section{Blow-Up Preserving SSP Reductions for Various Problems}\label{appendix:blow-up-preserving-ssp-reduction:reductions}

\begin{mdframed}[nobreak=true]
    \begin{description}
        \item[]\textsc{Set Cover}\hfill\\
        \textbf{Instances:} Sets $S_i \subseteq \fromto{1}{m}$ for $i \in \fromto{1}{n}$, number $k \in \N$.\\
        \textbf{Universe:} $\{S_i \mid i \in \fromto{1}{n}\} =: \U$.\\
        \textbf{Solution set:} The set of all $S \subseteq \bigcup_i S_i$ with $|S| \leq k$ such that $\bigcup_{s \in S} s = \fromto{1}{m}$.
    \end{description}
\end{mdframed}
The reduction from \textsc{Vertex Cover} to \textsc{Set Cover} by Karp \cite{DBLP:conf/coco/Karp72} is a blow-up preserving SSP reduction.
Let $G=(V,E)$ the graph and $k$ be the threshold of the \textsc{Vertex Cover} instance.
Each edge $e \in E$ is mapped to a unique number $n_e$ in $\fromto{1}{|E|}$.
Each vertex $v \in V$ is mapped into a set $S_v = \{n_e \mid e = \{v,w\} \in E\}$.
The threshold $k$ stays the same.
A vertex $v$ is in the \textsc{Vertex Cover} solution if and only if the corresponding set $S_v$ is in the \textsc{Set Cover} solution.
Then, we have $f(v) = S_v$ and $U_{\textit{on}} = U_{\textit{off}} = \emptyset$.

\begin{mdframed}[nobreak=true]
    \begin{description}   
        \item[]\textsc{Hitting Set}\hfill\\
        \textbf{Instances:} Sets $S_j \subseteq \fromto{1}{n}$ for $j \in \fromto{1}{m}$, number $k \in \N$.\\
        \textbf{Universe:} $\fromto{1}{n} =: \U$.\\
        \textbf{Solution set:} The set of all $H \subseteq \fromto{1}{n}$ with $|H| \leq k$ such that $H \cap S_j \neq \emptyset$ for all $j \in \fromto{1}{m}$.
    \end{description}
\end{mdframed}
Karp's reduction from \textsc{Vertex Cover} to \textsc{Hitting Set} \cite{DBLP:conf/coco/Karp72} is a blow-up preserving SSP reduction.
Let $G=(V,E)$ the graph and $k$ be the threshold of the \textsc{Vertex Cover} instance.
Each vertex $v \in V$ is mapped to a unique number $n_v \in \fromto{1}{|V|}$.
Each edge $e = \{v, w\} \in E$ is mapped to the set $S_e = \{n_v, n_w\}$.
The threshold $k$ stays the same.
A vertex $v$ is in the \textsc{Vertex Cover} solution if and only if the corresponding number $n_v$ is in the \textsc{Hitting Set} solution.
Then, we have $f(v) = n_v$ and $U_{\textit{on}} = U_{\textit{off}} = \emptyset$.

\begin{mdframed}[nobreak=true]
    \begin{description}
        \item[]\textsc{Feedback Vertex Set}\hfill\\
        \textbf{Instances:} Directed Graph $G = (V, A)$, number $k \in \N$.\\
        \textbf{Universe:} Vertex set $V =: \U$.\\
        \textbf{Solution set:} The set of all vertices $V'$ of size at most $k$ such that after deleting $V'$ from $G$, the resulting graph is cycle-free (i.e. a forest).
    \end{description}
\end{mdframed}
The reduction by Karp from \textsc{Vertex Cover} to \textsc{Feedback Vertex Set} \cite{DBLP:conf/coco/Karp72} is a blow-up preserving SSP reduction.
Let $G = (V, E)$ be the \textsc{Vertex Cover} instance and $G' = (V', A')$ the \textsc{Feedback Vertex Set} instance.
Every vertex $v \in V$ is mapped to a correspondent vertex $v' \in V'$.
Each edge $e = \{v,w\} \in E$ is mapped to two arcs $(v', w'), (w', v') \in A'$.
The threshold $k$ stays the same.
A vertex $v$ is in the \textsc{Vertex Cover} solution if and only if the corresponding vertex $v'$ is in the \textsc{Feedback Vertex Set} solution.
Then, we have $f(v) = v'$ and $U_{\textit{on}} = U_{\textit{off}} = \emptyset$.

\begin{mdframed}[nobreak=true]
    \begin{description}
        \item[]\textsc{Feedback Arc Set}\hfill\\
        \textbf{Instances:} Directed Graph $G = (V, A)$, number $k \in \N$.\\
        \textbf{Universe:} Arc set $A =: \U$.\\
        \textbf{Solution set:} The set of all arcs $A'$ of size at most $k$ such that after deleting $A'$ from $G$, the resulting graph is cycle-free (i.e. a forest).
    \end{description}
\end{mdframed}
A modification of a reduction by Karp from \textsc{Vertex Cover} to \textsc{Feedback Arc Set} \cite{DBLP:conf/coco/Karp72} is a blow-up preserving SSP reduction.
Let $G = (V, E)$ be the \textsc{Vertex Cover} instance and $G' = (V', A')$ the \textsc{Feedback Arc Set} instance.
Each vertex $v \in V$ is mapped to two vertices $v'_0, v'_1 \in V'$ as well as to an arc $(v'_0, v'_1) \in A'$.
Every edge $\{v, w\} \in E$ to $|V + 1|$ once subdivided arcs from $v_1$ to $w_0$ and to $|V + 1|$ once subdivided arcs from $w_1$ to  $v_0$.
The threshold $k$ stays the same.
A vertex $v$ is in the \textsc{Vertex Cover} solution if and only if the corresponding arc $(v'_0, v'_1)$ is in the \textsc{Feedback Arc Set} solution.
Then, we have $f(v) = (v_0, v_1)$, $U_{\textit{off}} = \{(v'_1, w'_0), (w'_1, v'_0) \mid e = \{v, w\} \in E\}$ and $U_{\textit{on}} = \emptyset$

\begin{mdframed}[nobreak=true]
    \begin{description} 
        \item[]\textsc{Uncapacitated Facility Location}\hfill\\
        \textbf{Instances:} Set of clients $C = \fromto{1}{m}$, set of potential facilities $F = \fromto{1}{n}$, fixed cost of opening facility function $f: F \rightarrow \Z$, service cost function $c: F \times C \rightarrow \Z$, cost threshold $k \in \Z$\\
        \textbf{Universe:} Facility set $F =: \U$.\\
        \textbf{Solution set:} The set of sets $F' \subseteq F$ s.t. $\sum_{i \in F'} f(i) + \sum_{j \in C} \min_{i \in F'} c(i, j) \leq k$.
    \end{description}
\end{mdframed}

\begin{samepage}
    \begin{mdframed}[nobreak=true]
    	\begin{description} 
        \item[]\textsc{p-Center}\hfill\\
        \textbf{Instances:} Set of clients $C = \fromto{1}{m}$, set of potential facilities $F = \fromto{1}{n}$, service cost function $c: F \times C \rightarrow \Z$, facility threshold $p \in \N$, cost threshold $k \in \Z$\\
        \textbf{Universe:} Facility set $F =: \U$.\\
        \textbf{Solution set:} The set of sets $F' \subseteq F$ s.t. $|F'| \leq p$ and $\max_{j \in C} \min_{i \in F'} c(i, j) \leq k$.
    	\end{description}
    \end{mdframed}
\end{samepage}

\begin{mdframed}[nobreak=true]
    \begin{description} 
        \item[]\textsc{p-Median}\hfill\\
        \textbf{Instances:} Set of clients $C = \fromto{1}{m}$, set of potential facilities $F = \fromto{1}{n}$, service cost function $c: F \times C \rightarrow \Z$, facility threshold $p \in \N$, cost threshold $k \in \Z$\\
        \textbf{Universe:} Facility set $F =: \U$.\\
        \textbf{Solution set:} The set of sets $F' \subseteq F$ s.t. $|F'| \leq p$ and $\sum_{j \in C} \min_{i \in F'} c(i, j) \leq k$.
    \end{description}
\end{mdframed}

Note that we define these problems explicitly as SSP and not as LOP because in the standard interpretation, the objective function is not linear.
The reduction by Cornuéjols, Nemhauser, and Wolsey \cite{cornuejols1983uncapicitated} from \textsc{Vertex Cover} to \textsc{Uncapacitated Facility Location} is a blow-up preserving SSP reduction.
The same reduction is also valid for \textsc{p-Center} and \textsc{p-Median}.
Let $G=(V,E)$ be the \textsc{Vertex Cover} instance and $(C, F, f, c)$ be the \textsc{Uncapacitated Facility Location} instance.
We map each vertex $v \in V$ to a facility $i_v \in F$.
Each edge $e \in E$ is mapped to a client $j_e \in C$.
We define $c(i_v, j_e) = 0$ if $v \in e$ and $c(i_v, j_e) = |V|+1$ otherwise.
At last, we set $f(i_v) = 1$ for all $i_v \in F$ and leave the parameter $k$ unchanged.
(For \textsc{p-Center} and \textsc{p-Median}, we use the equivalent constraint defined over the threshold $|F'| \leq p$ by setting $p = k$).
A vertex $v$ is in the \textsc{Vertex Cover} solution if and only if the corresponding facility $i_v$ is in the \textsc{Facility Location} solution.
Then, we have $f(v) = i_v$ and $U_{\textit{on}} = U_{\textit{off}} = \emptyset$.

\begin{mdframed}[nobreak=true]
    \begin{description}   
        \item[]\textsc{Clique}\hfill\\
        \textbf{Instances:} Graph $G = (V, E)$, number $k \in \N$.\\
        \textbf{Universe:} Vertex set $V =: \U$.\\
        \textbf{Solution set:} The set of all cliques of size at least $k$.
    \end{description}
\end{mdframed}
There is a reduction by Garey and Johnson from \textsc{Independent Set} to \textsc{Clique} \cite{DBLP:books/fm/GareyJ79}, which is a blow-up preserving SSP reduction.
Let $G = (V, E)$ be the \textsc{Independent Set} instance and $G' = (V', E')$ the \textsc{Clique} instance.
Every vertex $v \in V$ to a correspondent vertex $v' \in V'$.
Each edge $\{v, w\} \in E$ is mapped to a non-edge $\{v', w'\} \notin E'$ and every non-edge $\{v, w\} \notin E$ is mapped to an edge $\{v', w'\} \in E'$.
The threshold $k$ stays the same.
A vertex $v$ is in the \textsc{Independent Set} solution if and only if the corresponding vertex $v'$ is in the \textsc{Clique} solution.
Then, we have $f(v) = v'$ and $U_{\textit{on}} = U_{\textit{off}} = \emptyset$.

\begin{mdframed}[nobreak=true]
    \begin{description}   
        \item[]\textsc{Knapsack}\hfill\\
        \textbf{Instances:} Objects with prices and weights $\fromto{(p_1, w_1)}{(p_n, w_n)} \subseteq \N^2$, and $W, P \in \N$.\\
        \textbf{Universe:} $\fromto{(p_1, w_1)}{(p_n, w_n)} =: \U$.\\
        \textbf{Solution set:} The set of all $S \subseteq \U$ with $\sum_{(p_i, w_i) \in S} p_i \geq P$ and $\sum_{(p_i, w_i) \in S}w_i \leq W$.
    \end{description}
\end{mdframed}
The reduction from \textsc{Subset Sum} to \textsc{Knapsack} is a simple folklore result.
Each number $a$ is transformed to an object $(a,a)$ of the same price and weight.
The price and weight threshold $P$ and $W$ is set to the target value $M$ of the \textsc{Subset Sum} instance.
A number $a$ is in the \textsc{Subset Sum} solution if and only if the corresponding object $(a,a)$ is in the \textsc{Knapsack} solution.
Then, we have $f(a) = (a,a)$ and $U_{\textit{on}} = U_{\textit{off}} = \emptyset$.

\begin{mdframed}[nobreak=true]
    \begin{description}
        \item[]\textsc{Partition}\hfill\\
        \textbf{Instances:} Numbers $\fromto{a_1}{a_n} \subseteq \N$.\\
        \textbf{Universe:} $\fromto{a_1}{a_n} =: \U$.\\
        \textbf{Solution set:} The set of all sets $S \subseteq \U$ with $\sum_{a_i \in S}a_i = \sum_{a_j \notin S}a_j$.
    \end{description}
\end{mdframed}
For a reduction between \textsc{Subset Sum} and \textsc{Partition}, we use basically the same reduction as Karp's from \textsc{Knapsack} to \textsc{Partition} \cite{DBLP:conf/coco/Karp72}.
Each number $a$ is mapped to the same number $a$ in the \textsc{Partition} instance.
Furthermore, additional numbers $M+1$ and $\sum_{i} a_i + 1 - M$ are added.
A number $a$ is in the \textsc{Subset Sum} solution if and only if the corresponding number $a$ is in the \textsc{Knapsack} solution.
Then, we have $f(a) = a$, $U_{\textit{on}} = \{\sum_i a_i + 1 - M\}$, and $U_{\textit{off}} = \{ M + 1 \}$.

\begin{mdframed}[nobreak=true]
    \begin{description}
        \item[]\textsc{Two Machine Scheduling}\hfill\\
        \textbf{Instances:} Jobs with processing time $\fromto{(t_1)}{(t_n)} \subseteq \N$, threshold $T \in \N$.\\
        \textbf{Universe:} The set of jobs $\fromto{(t_1)}{(t_n)} =: \U$.\\
        \textbf{Solution set:} The set of all $J_1 \subseteq \U$ such that $\sum_{t_i \in J_1}t_i \leq T$ and $\sum_{t_j \in J_2}t_j \leq T$ with $J_2 = \U \setminus J_1$, i.e. both machines finish in time $T$.
    \end{description}
\end{mdframed}
The reduction from \textsc{Partition} to \textsc{Two-Machine-Scheduling} is a folklore reduction, which exploits the equivalence of the problems.
Each number $a$ in the \textsc{Partition} instance is mapped to a job with processing time $a$ in the \textsc{Two-Machine-Scheduling} instance.
The threshold is set to $T = \frac{1}{2} \sum_i a_i$.
A number $a$ is in the \textsc{Partition} solution if and only if the corresponding job with processing time $a$ is in the \textsc{Two-Machine-Scheduling} solution.
Then, we have $f(a) = a$, $U_{\textit{on}} = U_{\textit{off}} = \emptyset$.

\begin{mdframed}[nobreak=true]
    \begin{description}
        \item[]\textsc{Directed Hamiltonian Cycle}\hfill\\
        \textbf{Instances:} Directed Graph $G = (V, A)$.\\
        \textbf{Universe:} Arc set $A =: \U$.\\
        \textbf{Solution set:} The set of all sets $C \subseteq A$ forming a Hamiltonian cycle.
    \end{description}
\end{mdframed}
We extend the reduction from \textsc{3Sat} to \textsc{Directed Hamiltonian Path} from Arora and Barak by simply adding an arc from $t$ to $s$.
An arc $a \in A$ is in the solution \textsc{Directed Hamiltonian Path} if and only if the corresponding arc is in the \textsc{Directed Hamiltonian Cycle} solution.
Then, we have $f(a) = a$, $U_{\textit{on}} = \{(t,s)\}$, and $U_{\textit{off}} = \emptyset$.

\begin{mdframed}[nobreak=true]
    \begin{description}  
        \item[]\textsc{Undirected Hamiltonian Cycle}\hfill\\
        \textbf{Instances:} Graph $G = (V, E)$.\\
        \textbf{Universe:} Edge set $E =: \U$.\\
        \textbf{Solution set:} The set of all sets $C \subseteq E$ forming a Hamiltonian cycle.
    \end{description}
\end{mdframed}
Karp's reduction from \textsc{Directed Hamiltonian Cycle} to \textsc{Undirected Hamiltonian Cycle} \cite{DBLP:conf/coco/Karp72} is a blow-up preserving SSP reduction.
Let $G = (V, A)$ be the \textsc{Directed Hamiltonian Cycle} and $G' = (V', E')$ be the \textsc{Undirected Hamiltonian Cycle} instance.
Each vertex $v \in V$ is mapped to three vertices $v'_{in}, v', v'_{out}$ and edges $\{v'_{in}, v'\}, \{v', v'_{out}\}$.
Each arc $(v, w) \in A$ is mapped to an edge $\{v'_{out}, w'_{in}\} \in E$.
An arc $a \in A$ is in the solution \textsc{Directed Hamiltonian Path} if and only if the corresponding edge $\{v'_{out}, w'_{in}\} \in E$ is in the \textsc{Undirected Hamiltonian Cycle} solution.
Then, we have $f((v, w)) = \{v'_{out}, w'_{in}\}$, $U_{\textit{on}} = \{ \{v'_{in}, v'\}, \{v', v'_{out} \} \mid v \in V\}$, and $U_{\textit{off}} = \emptyset$.

\begin{mdframed}[nobreak=true]
    \begin{description}
        \item[]\textsc{Traveling Salesman Problem}\hfill\\ \label{apx:ssp-reduction:tsp}
        \textbf{Instances:} Complete Graph $G = (V, E)$, weight function $w: E \rightarrow \Z $, number $k \in \N$.\\
        \textbf{Universe:} Edge set $E =: \U$.\\
        \textbf{Feasible solution set:} The set of all TSP tours $T\subseteq E$.\\
        \textbf{Solution set:} The set of feasible $T$ with $w(T) \leq k$.
    \end{description}
\end{mdframed}
There is an folklore reduction from \textsc{Undirected Hamiltonian Cycle} to \textsc{Traveling Salesman Problem}, which is a blow-up preserving SSP reduction.
Let $G = (V, E)$ be the \textsc{Undirected Hamiltonian Cycle} instance and $G' = (V', E', w', k')$ the \textsc{Traveling Salesman Problem} instance.
Every vertex $v \in V$ is mapped to a corresponding vertex $v' \in V'$.
Then, all vertices are connected to form a complete graph.
The weight function $w': E' \rightarrow \Z$ is defined for all $\{v',w'\} \in E'$ as
$$
    w(\{v',w'\}) = \begin{cases}
        0, \quad \text{if} \ \{v,w\} \in E\\
        1, \quad \text{if} \ \{v,w\} \notin E
    \end{cases}
$$
At last, we set $k = 0$.
An edge $\{v,w\} \in E$ is in the solution \textsc{Undirected Hamiltonian Path} if and only if the corresponding edge $\{v',w'\} \in E'$ is in the \textsc{Traveling Salesman Problem} solution.
Then, we have $f(e) = e'$, $U_{\textit{on}} = \emptyset$ and $U_{\textit{off}} = \{\{v',w'\} \mid \{v,w\} \notin E\}$.

\begin{mdframed}[nobreak=true]
    \begin{description}
        \item[]\textsc{Directed} $k$-\textsc{Disjoint Directed Path}\hfill\\
        \textbf{Instances:} Graph $G = (V, A)$, $s_i, t_i \in V$ for $i \in \fromto{1}{k}$.\\
        \textbf{Universe:} Arc set $A =: \U$.\\
        \textbf{Solution set:} The set of all disjoint paths $(P_1, \ldots, P_k)$ such that $P_i$ is the path from $s_i$ to $t_i$ for $i \in \fromto{1}{k}$.
    \end{description}
\end{mdframed}
The following reduction from \textsc{Directed Two Disjoint Path} to \textsc{Directed} $k$-\textsc{Disjoint Directed Path} is a blow-up preserving SSP reduction.
Each vertex $v \in V$ respectively each arc $a \in A$ is mapped to a corresponding vertex $v' \in V'$ respectively arc $a' \in A'$.
Additionally, we add $k-2$ additional vertex pairs $s_i, t_i$ for $i \in \fromto{3}{k}$ connected by arcs $(s_i, t_i)$ for all $i \in \fromto{3}{k}$.
An arc $a \in A$ is in the solution \textsc{Directed Two Disjoint Path} if and only if the corresponding arc $a' \in A'$ is in the \textsc{Directed} $k$-\textsc{Disjoint Directed Path} solution.
Then, we have $f(a) = a'$ and $U_{\textit{on}} = \{(s_i, t_i) \mid i \in \fromto{3}{k}\}$, and $U_{\textit{off}} = \emptyset$.

\end{document}